\documentclass[11pt,a4paper,reqno]{amsart}
\usepackage[english]{babel}
\usepackage[utf8]{inputenc} 
\usepackage{fancyhdr}
\usepackage{indentfirst}
\usepackage{color}     
\usepackage{graphicx}   
\usepackage{newlfont}

\DeclareMathAlphabet{\dutchcal}{U}{dutchcal}{m}{n}
\SetMathAlphabet{\dutchcal}{bold}{U}{dutchcal}{b}{n}
\DeclareMathAlphabet{\dutchbcal} {U}{dutchcal}{b}{n}

\usepackage{commath}
\usepackage{amsfonts}
\usepackage{amsmath}
\usepackage{amssymb}
\usepackage{amsthm}
\usepackage{latexsym}
\usepackage{mathrsfs}
\usepackage{mathtools}
\mathtoolsset{showonlyrefs=true}
\usepackage{verbatim}
\usepackage[all]{xy}
\usepackage{lineno}
\usepackage{leftidx}
\usepackage[hidelinks]{hyperref}
\usepackage{units}
\usepackage{xfrac}
\usepackage{faktor}
\usepackage{float}
\usepackage{tikz}
\usetikzlibrary{arrows, positioning, automata}

\usepackage{chngcntr}
\usepackage{apptools}

\usepackage[flushleft]{threeparttable}

\theoremstyle{plain} 
\newtheorem*{theo*}{Theorem}
\newtheorem*{cor*}{Corollary}
\newtheorem*{con*}{Conjecture}
\newtheorem{theo}{Theorem}[section] 
\newtheorem{prop}[theo]{Proposition}
\newtheorem{cor}[theo]{Corollary}
\newtheorem{lem}[theo]{Lemma}
\newtheorem{con}[theo]{Conjecture}

\theoremstyle{definition}

\theoremstyle{definition}
\newtheorem{defin}[theo]{Definition}
\newtheorem{ex}[theo]{Example}

\theoremstyle{remark}
\newtheorem{rem}[theo]{Remark}
\newtheorem*{rem*}{Remark}
\newtheorem*{cap*}{Caption}

\AtAppendix{\counterwithin{theo}{section}}

\newcommand{\Z}{\mathbb{Z}}
\newcommand{\Zplus}{\Z_{>0}}
\newcommand{\Zpluseq}{\Z_{\geq0}}
\newcommand{\R}{\mathbb{R}}
\newcommand{\C}{\mathbb{C}}

\newcommand{\pSone}{S^1\backslash\{-1\}}

\newcommand{\Mob}{\operatorname{M\ddot{o}b}}
\newcommand{\Diff}{\operatorname{Diff}}
\newcommand{\Vir}{\mathfrak{Vir}}

\newcommand{\Uone}{\operatorname{U}(1)}
\newcommand{\g}{\mathfrak{g}} 
\newcommand{\h}{\mathfrak{h}}

\newcommand{\End}{\operatorname{End}}

\newcommand{\J}{\mathcal{J}}
\newcommand{\A}{\mathcal{A}}
\newcommand{\B}{\mathcal{B}}
\newcommand{\Virnet}{\operatorname{Vir}}
\newcommand{\SVirnet}{\operatorname{SVir}}
\newcommand{\F}{\mathcal{F}}

\newcommand{\parzero}{{\overline{0}}}
\newcommand{\parone}{{\overline{1}}}
\newcommand{\half}{\frac{1}{2}}
 
\newcommand{\bos}{\dutchcal{b}}
\newcommand{\fer}{\dutchcal{f}}

\newcommand{\scalar}{(\cdot|\cdot)}

\newcommand{\bilinear}{(\cdot\,,\cdot)}

\begin{document}
	
\author{Sebastiano Carpi}
	\address{Dipartimento di Matematica, Universit\`a di Roma ``Tor Vergata'', Via della Ricerca Scientifica, 1, 00133 Roma, Italy\\
		E-mail: {\tt carpi@mat.uniroma2.it}
	}
\author{Tiziano Gaudio}
	\address{Dipartimento di Matematica, Universit\`a di Roma ``Tor Vergata'', Via della Ricerca Scientifica, 1, 00133 Roma, Italy\\
		E-mail: {\tt gaudio@mat.uniroma2.it}
	}

\title[Conformal nets from minimal $W$-algebras]{Conformal nets from minimal $W$-algebras
}

\date{31 October 2025}

\begin{abstract}
We show the strong graded locality of all unitary minimal $W$-algebras, so that they give rise to irreducible graded-local conformal nets. Among these unitary vertex superalgebras, up to taking tensor products with free fermion vertex superalgebras, there are the unitary Virasoro vertex algebras ($N=0$) and the unitary $N=1,2,3,4$ super-Virasoro vertex superalgebras. Accordingly, we have a uniform construction that gives, besides the already known $N=0,1,2$ super-Virasoro nets, also the new $N=3,4$ super-Virasoro nets.
All strongly rational unitary minimal $W$-algebras give rise to previously known completely rational graded-local conformal nets and we conjecture that the converse is also true. We prove this conjecture for all unitary $W$-algebras corresponding to the $N=0,1,2,3,4$ super-Virasoro vertex superalgebras.
 
\end{abstract}
	
\maketitle

\tableofcontents

\section{Introduction} 

Quantum field theory (QFT) is a central branch of theoretical physics that has developed in an attempt to give a quantum description of systems with an infinite number of degrees of freedom.  It is well known that the quantization of interacting relativistic fields naturally gives rise to divergent quantities. In favorable cases, they can be cured by the procedure of renormalization, but often only from a perturbative point of view.  This naturally leads to the demand for axiomatic descriptions of QFT, so as to provide it with a mathematically rigorous framework that allows a deeper understanding of the theory.  Two of the main axiomatic approaches to QFT are the Wightman one and algebraic quantum field theory (AQFT), see \cite{SW64,GJ87,BW92,Str93,Haag96,Ara10}. Although the connection between these two axiomatic settings has 
been studied for long time, their mathematical relations are not yet completely understood.   

Conformal field theory (CFT) in two space-time dimensions describes a class of QFTs that are quite specific. However, they give rise to an impressive number of connections between various fields of physics and mathematics: from critical phenomena to string theory; from finite groups and modular forms to subfactors and noncommutative geometry, see \cite{DMS97,EK98,Gan06,Kaw15}. Chiral CFTs are CFTs in two space-time dimensions which are generated by fields depending on a single light-ray coordinate (the same for all the generating fields), so that they can be considered as CFTs on the compactified light-ray $S^1$. They are the building blocks of  two-dimensional CFTs. In the AQFT setting, they are described in terms of conformal nets on $S^1$ or more generally M\"{o}bius covariant nets on $S^1$. Besides the conformal net and the Wightman axioms, a third mathematical axiomatization for chiral CFTs, the one based on vertex algebras \cite{FLM88,FHL93,Kac01}, is nowadays quite popular and very well established. 
The latter axiomatization can be considered as a purely algebraic variant of the Wightman axioms, but it shares with the conformal net setting a strong emphasis on representation theory aspects.  

A systematic study of the connection between these different axiomatic settings for chiral CFTs started with the work of Carpi, Kawahigashi, Longo and Weiner \cite{CKLW18} on the construction of local conformal nets from suitably well-behaved vertex operator algebras (VOAs), called strongly local. The strategy in \cite{CKLW18} can be described in the following way. Assuming certain polynomial energy bounds for a unitary VOA $V$, one can define operator-valued distributions acting on the Hilbert space completion $\mathcal{H}_V$ of $V$ from the formal distributions given by the vertex operators of $V$. These operator-valued distributions satisfy a chiral CFT version of the Wightman axioms. Following a traditional path, they can be used to define a net of von Neumann algebras. In order to get a conformal net, the locality axiom for the latter is the only missing point. The strong locality condition means exactly that, besides the already assumed polynomial energy bounds, the net is local and hence a local conformal net. Many important examples of unitary VOAs have been shown to be strongly local in \cite{CKLW18}. Moreover, it has been conjectured in the same paper that every unitary VOA is strongly local. 

The correspondence in \cite{CKLW18} has been recently extended to unitary vertex operator superalgebras (VOSAs) and graded-local conformal nets in \cite{CGH} in order to allow the presence of fermionic fields that play an important role in many constructions. Following the strategy in \cite{CKLW18}, a suitable class of unitary VOSAs called strongly graded-local has been defined in \cite{CGH}. Then the corresponding graded-local conformal nets are defined through the graded-local family of Wightman fields on $S^1$ associated with the vertex operators. Moreover, in the same paper, the strong graded locality is proved for many examples of unitary VOSAs. As a consequence, many of the already known graded-local conformal nets on $S^1$ are given back together with a remarkable family of new nets. Further remarkable examples of strongly graded-local VOSAs have been recently given in \cite{Gau}.

It should be pointed out that the strong graded locality is not needed in order to define conformal or M\"{o}bius covariant Wightman theories on $S^1$ from unitary VOSAs, see \cite{RTT22}.  Actually, one can also drop the unitarity condition by generalizing the notion of Wightman CFTs to the non-unitary setting, see \cite{CRTT25}. On the other hand, conformal nets act on Hilbert spaces and they are intrinsically unitary. Accordingly, one has to require that the associated VOSAs are also unitary. Concerning the polynomial energy bounds, they play an important role in the known constructions of conformal nets on $S^1$, although they are expected to follow from unitarity. Finally, proving the strong graded locality is essential because the graded locality for the nets of von Neumann algebras is presently not known to follow from the corresponding graded locality for the Wightman fields. Actually, this problem is a central issue in the lack of a complete understanding of the connection of Wightman QFT with AQFT.

In this article, we prove that all unitary minimal $W$-algebras are strongly graded-local and hence they define a corresponding family of irreducible graded-local conformal nets on $S^1$. On the one hand, some of these nets, such as the $N=0,1,2$ super-Virasoro nets, where already known by direct construction from the unitary vacuum representations of the corresponding superconformal algebras. On the other hand, many of them appear to be completely new. In particular, we can construct the new $N=3,4$ super-Virasoro nets.   
 
 $W$-algebras are a remarkable family of vertex superalgebras. The paradigmatic family of examples is given by the Zamolodchikov $W_3$-algebras. 
 Unitary $W_3$-algebras have been classified in \cite{CTW23} and proved to be strongly local in \cite{CTW22, CTW23b}. An important family of $W$-algebras 
 is the one of universal affine $W$-algebras $W^k(\mathfrak{g},x,f)$ that is obtained by quantization of the classical Hamiltonian reduction, and their simple quotients $W_k(\mathfrak{g},x,f)$, see \cite{FF90, FF90b, KRW03, KW04, DK06}. Here 
 $\mathfrak{g}$ is a simple finite-dimensional Lie superalgebra with even elements $x$ and $f$ satisfying suitable properties. 
 Through this procedure of quantization, one gets the so called principal $W$-algebras when $f$ is a principal nilpotent element, whereas minimal $W$-algebras correspond to minimal nilpotent elements. 
 The $W_3$-algebras can be identified with principal $W$-algebras  corresponding to $\mathfrak{g}= \mathfrak{sl}_3$. The unitary ones with central charge $c <2$ correspond to the so called discrete series. More generally, for $\mathfrak{g}$ of $ADE$ type the discrete series principal $W$-algebras can be shown to be unitary and strongly local as a consequence of their coset realization, see \cite{ACL19,Ten24, Gui}.   On the other hand, 
there seems to be no complete classification result for unitary principal $W$-algebras besides the  cases $\mathfrak{sl}_2$ and $\mathfrak{sl}_3$, corresponding to the Virasoro vertex algebra and the $W_3$-algebras respectively, see \cite{CTW23}. The situation appears to be simpler in the case of minimal $W$-algebras. In fact, the unitary minimal $W$-algebras have been recently classified in \cite{KMP23} and their unitary representation theory has been studied in detail in \cite{KMP23,AKMP24,KMP25}. 

In the present paper, beside the proof of the strong graded locality, we discuss some aspects of the representation theory of the corresponding graded-local conformal nets. It turns out that all the strongly rational unitary minimal $W$-algebras give rise to 
graded-local conformal nets that are completely rational in the sense of \cite{KLM01}. In the converse direction, we show that the $N=0,1,2,3,4$ super-Virasoro nets corresponding to non-rational minimal $W$-algebras are not completely rational. We conjecture that a unitary minimal $W$-algebra is strongly rational if an only if the corresponding conformal net is completely rational.

This paper is organized as follows. In Section \ref{section:preliminaries}, we discuss some preliminaries on unitary VOSAs and the corresponding graded-local conformal nets. In Section \ref{section:strong_graded_locality}, we state and prove our strong graded locality result for all unitary minimal $W$-algebras. In Section \ref{section:representations}, we discuss with some detail the notion of complete rationality for graded-local conformal nets. Then we move to prove that the graded-local conformal nets associated with the non-rational minimal $W$-algebras corresponding to the 
$N=0,1,2,3,4$ super-Virasoro algebras are not completely rational. This is achieved by proving the existence of infinitely many equivalence classes of irreducible locally normal representations of these nets.  In the same section, we show that the main result in \cite{MTW18} on the split property for local conformal nets can be generalized to the graded-local case thanks to the results in \cite{Dop82}.

\section{Preliminaries}  \label{section:preliminaries}

	We refer to \cite{Kac01}, see also \cite{FLM88, FHL93, Xu98, LL04}, for the basics of the theory of vertex superalgebras and to \cite[Section 3]{CGH}, see also \cite[Section 4]{CKLW18}, for notations and basic definitions.
	
	For a \textit{vertex superalgebra} $V=V_\parzero\oplus V_\parone$, we use $\Gamma_V$ for the \textit{parity operator} of $V$, that is $\Gamma_V(a)=(-1)^pa$ whenever $a\in V_{\overline{p}}$ with $p\in\{0,1\}$, where $p(a):=\overline{p}\in\Z/2\Z$ is called the \textit{parity} of $a$. Accordingly,  vectors in $V_\parzero$ are called \textit{even}, whereas vectors in $V_\parone$ are called \textit{odd}. 
	$T$ is the \textit{infinitesimal translation operator} and $Y(a,z)=\sum_{n\in\Z}a_{(n)}z^{-n-1}$ denotes the \textit{vertex operator} associated to a vector $a\in V$. As usual, $\Omega\in V_\parzero$ denote the \textit{vacuum vector}.
	A \textit{Virasoro vector} $\nu\in V_\parzero$ is a vector whose coefficients of the corresponding field $Y(\nu, z)=\sum_{m\in\Z}L_mz^{-m-2}$ satisfy the \textit{Virasoro algebra $\Vir$ commutation relations} for a given \textit{central charge} $c\in\C$:
	\begin{equation} \label{eq:virasoro_cr}
		\forall n,m\in\Z\qquad
		[L_n,L_m]=(n-m)L_{n+m}+\frac{c(n^3-n)}{12}\delta_{n,-m}1_V \,.
	\end{equation}
	If $L_{-1}=T$ and $L_0$ is diagonalizable on $V$, then $\nu$ is called a \textit{conformal vector}. Accordingly, $V$ endowed with a conformal vector $\nu$ is called a \textit{conformal vertex superalgebra}.
	Let $V_n:=\mathrm{Ker}(L_0-n1_V)$ for all $n\in \C$ be the $L_0$-eigenspaces. 
	If $a\in V_n$ for some $n\in\C$, then it is called  \textit{homogeneous} of \textit{conformal weight} $d_a:=n$ and we use to write $Y(a,z)=\sum_{n\in \Z-d_a}a_nz^{-n-d_a}$, where $a_n:=a_{(n+d_a-1)}$ if $n\in\Z-d_a$ and $a_n=0$ otherwise. 
	In this way, one can write $Y(z^{L_0}a, z)=\sum_{n\in\C}a_nz^{-n}$ for all $a\in V$.
	Moreover, a homogeneous vector $a$ is said to be \textit{primary} if $L_na=0$ for all $n>0$ and \textit{quasi-primary} if $L_1a=0$. Therefore, $\Omega$ is a primary vector, whereas $\nu$ is quasi-primary.
	A \textit{vertex operator superalgebra (VOSA)} is a conformal vertex superalgebra $V$ where the corresponding conformal vector $\nu$ satisfies: for all $n\in\half\Z$, $V_n$ is finite-dimensional;
	\begin{equation}
		V_\parzero=\bigoplus_{n\in\Z}V_n 
		\,,\qquad
		V_\parone=\bigoplus_{n\in\Z-\half}V_n 
	\end{equation}
	and there exists an integer $N$ such that $V_n=\{0\}$ for all $n< N$.
	If $V_\parone=\{0\}$, then we also use the terminology \textit{vertex algebra}, \textit{conformal vertex algebra} and \textit{vertex operator algebra (VOA)} respectively.
	
	\begin{rem}
		An alternative, although equivalent, definition of vertex superalgebra is given in terms of the \textit{$\lambda$-bracket} $[\cdot_\lambda \cdot]$ in place of the usual $(n)$-product, see \cite[Section 1.5]{DK06} for details, which is defined by:
		\begin{equation} \label{eq:lambda_bracket}
			\forall a,b\in V\qquad
			[a_\lambda b]:= \sum_{j=0}^{+\infty}\frac{\lambda^j}{j!} (a_{(j)}b)
		\end{equation}
		where $\lambda$ is indeterminate.
		In this paper, given a $\lambda$-bracket $[a_\lambda b]$ as in \eqref{eq:lambda_bracket}, we will use the $\lambda$-coefficients $\frac{(a_{(j)}b)}{j!}$ to calculate commutators of type $[a_p,b_q]$ for $p,q\in\half\Z$.
		To this aim, we recall the following form of the \textit{Borcherds commutator formula} \cite[Proposition 4.8]{Kac01} in the conformal case, see also \cite[Eq.\ (15)]{CT23}: 
		\begin{equation} \label{eq:borcherds_commutator_formula_weighted}
			\forall p, q\in\C \,\,\,\forall c\in V \qquad
			[a_p,b_q]c = \sum_{j=0}^{+\infty}\binom{p+d_a-1}{j} (a_{(j)}b)_{p+q} c
			\,.
		\end{equation}
	\end{rem} 
	
	A VOSA $V$ is said to be \textit{unitary} if, see \cite[Section 3.1]{CGH}, cf.\ also \cite[Section 2.1]{AL17}, there exists a \textit{scalar product} (that we assume to be linear in the second variable) $\scalar$ on $V$ such that: it is \textit{normalized}, that is $(\Omega|\Omega)=1$; it is \textit{invariant}, that is there exists an antilinear VOSA involution $\theta$ of $V$, called the \textit{PCT operator}, satisfying
	\begin{equation}
		\forall a,b,c\in V\qquad
		(Y(\theta(a),z)b|c)=(b|Y(e^{zL_1}(-1)^{2L_0^2+L_0}z^{-2L_0}a, z^{-1})c)
		\,.
	\end{equation}   
	The above definition of unitarity is equivalent to the following, which will be mostly used in this paper, see \cite[Theorem 3.31]{CGH}, cf.\ also \cite[Section 2]{CT23}: there exists a normalized scalar product and an antilinear vector superspace involution $V\ni a\to \overline{a}\in V$ such that $\overline{\nu}=\nu$ and 
	\begin{equation}
		\forall a,b,c\in V \,\,\,\forall n\in\half\Z \qquad
		(a_nb|c)=(b|\overline{a}_{-n}c)
		\,.
	\end{equation}
	In this case, we have that $\overline{\Omega}=\Omega$ and there exists a unique antilinear VOSA automorphism $\theta$ such that $\overline{a}=e^{L_1}(-1)^{2L_0^2+L_0}\theta(a)$ for all $a\in V$. Moreover, $\theta$ turns out to be the PCT operator of $V$. 
	We say that $a\in V$ is \textit{Hermitian} if $\overline{a}=a$.
	We note that for any two homogeneous vectors $a$ and $b$ of any unitary VOSA, we have that $(a|b)=0$ whenever $d_a\not= d_b$. Moreover, a unitary VOSA $V$ is \textit{simple}, that is there are no non-trivial ideals, if and only if it is of \textit{CFT type}, that is $V_0=\C\Omega$ and $V_n=\{0\}$ for all $n<0$, see \cite[Proposition 3.10]{CGH}. 
	If $V$ is a simple unitary VOSA, then its unitary structure is unique up to unitary VOSA isomorphism, see \cite[Proposition 3.14]{CGH}. 
	Recall that the even part $V_\parzero$ of $V$ is a \textit{unitary} VOA in the sense of \cite{DL14} and \cite[Section 5]{CKLW18}, and it is also a \textit{unitary subalgebra} of $V$, see \cite[Section 3.5]{CGH}.

We recall the definition of a \textit{graded-local M\"{o}bius covariant net} $\A$ from \cite[Section 2]{CGH}, based on \cite[Section 2]{CKL08}, \cite[Section 2]{CHKLX15} and \cite[Section 2]{CHL15}.
This is an \textit{isotonous}, that is inclusion-preserving, map from the set $\J$ of \textit{intervals} (non-dense connected open subsets) of the unit circle $S^1$ to a family of von Neumann algebras $\A(I)$ with $I\in\J$, acting on a separable Hilbert space $\mathcal{H}$. 
The \textit{M\"{o}bius covariance} requires the existence of a \textit{strongly continuous unitary representation} $U$ of the universal cover $\Mob(S^1)^{(\infty)}$ of the \textit{M\"{o}bius group} $\Mob(S^1)$ of $S^1$ acting covariantly on $\A$. This means that for all $I\in\J$ and all $A\in\A(I)$, it holds that $U(\gamma)\A(I)U(\gamma)^*=\A(\dot{\gamma} I)$, where $\dot{\gamma}$ is the image under the covering map $\tilde{p}:\Mob(S^1)^{(\infty)}\to \Mob(S^1)$ of $\gamma$. 
We also assume that: $U$ is \textit{positive-energy}, that is the generator $L_0$ of the lift to $\Mob(S^1)^{(\infty)}$ of the one-parameter subgroup of rotations $\R\ni t\mapsto U(r^{(\infty)}(t))$ is positive;
there exists a \textit{vacuum}, that is a $U$-invariant vector $\Omega\in\mathcal{H}$, which is also cyclic for the von Neumann algebra $\A(S^1)$ generated by $\bigcup_{I\in\J} \A(I)$. 
The \textit{twisted} or \textit{graded locality} is realized by a self-adjoint unitary operator $\Gamma$ on $\mathcal{H}$, called the \textit{grading unitary}, such that $\A(I')\subseteq Z\A(I)'Z^*$, where $Z:=\frac{1_\mathcal{H}-i\Gamma}{1-i}$. 
We say that $\A$ is a \textit{graded-local conformal net} if it is also \textit{diffeomorphism covariant}: there exists an extension of $U$, denoted by the same symbol, to a \textit{positive-energy strongly continuous projective unitary representation} of the universal cover $\Diff^+(S^1)^{(\infty)}$ of the group $\Diff^+(S^1)$ of \textit{orientation-preserving diffeomorphisms} of $S^1$ such that
\begin{equation}
	\begin{split}
		\forall \gamma\in\Diff^+(S^1)^{(\infty)} \,\,\,\forall I\in\J \qquad
		U(\gamma)\A(I)U(\gamma)^*
			&=
		\A(\dot{\gamma} I) \\
		\forall\gamma\in \Diff(I)^{(\infty)} \,\,\,\forall I\in\J \,\,\,\forall A\in\A(I')	\qquad
		U(\gamma)AU(\gamma)^*
			&= A 
	\end{split}
\end{equation}
where: $\dot{\gamma}$ is the image under the covering map from $\Diff^+(S^1)^{(\infty)}$ to $\Diff^+(S^1)$ of $\gamma$, still denoted by $\tilde{p}$; whereas $\Diff(I)^{(\infty)}$ is the connected component to the identity of the pre-image under $\tilde{p}$ of the subgroup of those diffeomorphisms acting as the identity map on $I'$.
A graded-local M\"{o}bius covariant or conformal net is said to be \textit{irreducible} if $\Omega$ is unique, up to a constant factor, among the vectors in $\mathcal{H}$ that are invariant for the action of $\Mob(S^1)^{(\infty)}$ given by $U$. This is equivalent to every $\A(I)$ to be a $III_1$ factor. 
Finally, note that $e^{i2\pi L_0}=U(r^{(\infty)}(2\pi))=\Gamma$, so that $U$ factors through a representation of the double cover $\Diff^+(S^1)^{(2)}$. We denote by $\A^\Gamma$ the \textit{Bose subnet} of $\A$, that is the fixed point subnet of $\A$ with respect to the adjoint action of $\Gamma$. See \cite[Section 2.3]{CGH} for the definition of \textit{covariant subnets} and their related topics. 
If $\A$ is irreducible, then $\A^\Gamma$ can be considered as an \textit{irreducible local M\"{o}bius covariant} or \textit{local conformal net} acting on the $\Gamma$-invariant subspace $\mathcal{H}^\Gamma$ of $\mathcal{H}$. The \textit{Virasoro subnet} $\Virnet_c$ of a graded-local conformal net $\A$ is always a covariant subnet of $\A^\Gamma$, see e.g.\ \cite{Car04, KL04, CGH}.

Therefore, from a simple unitary VOSA $V$, it is possible to define an irreducible graded-local conformal net $\A_V$, if certain analytic assumptions are satisfied. In the following, such analytic assumptions are briefly summarized, see \cite[Section 4]{CGH} for a detailed treatment.

\begin{defin}  \label{defin:energy_bounds}
	Let $(V,\scalar)$ be a unitary VOSA. A vector $a\in V$ is said to satisfy \textit{$k$-th order (polynomial) energy bounds} for some non-negative real number $k$, if there exist non-negative real numbers $M$ and $s$ such that
	\begin{equation}
		\forall n\in\half\Z
		\,\,\,\forall b\in V\qquad
		\norm{a_nb}\leq M(1+\abs{n})^s\norm{(L_0+1_V)^kb}
	\end{equation}
	where $\norm{a}:=\sqrt{(a|a)}$ for all $a\in V$. If $k=1$, $a$ is said to satisfy \textit{linear energy bounds}; whereas if $k$ is not specify, $a$ is simply said to satisfy \textit{energy bounds}. Accordingly, $V$ is said to be \textit{energy-bounded} if every element satisfies energy bounds. 
\end{defin}

We shall need the following proposition:

\begin{prop}  \label{prop:energy_bounds_odd_vectors_anticommutator}
	Let $V$ be a unitary VOSA and let $b\in V_\parone$ be any odd vector. Suppose that there exist some non-negative real numbers $M, s$ and $k$ such that
	\begin{equation}
		\forall n\in\Z-\frac12
		\,\,\,\forall c\in V \qquad
		\norm{[b_n,\overline{b}_{-n}]c} \leq
		M(\abs{n}+1)^s\norm{(L_0+1_V)^kc}
		\,.
	\end{equation}
Then $b$ satisfies the following $\frac{k}{2}$-th order energy bounds:
\begin{equation}
	\forall n\in\Z-\frac12
	\,\,\,\forall c\in V \qquad
	\norm{b_nc} \leq
	\sqrt{M}(\abs{n}+1)^\frac{s}{2}\norm{(L_0+1_V)^\frac{k}{2}c}
	\,.
\end{equation}
\end{prop}

\begin{proof}
	First, note that for all $x,y\in V$ and all $n\in\half\Z$, the endomorphism $x_ny_{-n}$ commutes with $L_0$.
	Let $c$ be any vector in $V$. Then easy calculations together with the Cauchy-Schwarz inequality in the fifth row do the job:
	for all $n\in\Z-\half$,
	\begin{equation}
		\begin{split}
			\norm{b_nc}^2
			 &\leq
			\norm{\overline{b}_{-n}c}^2 + \norm{b_nc}^2 
			 =
			 (\overline{b}_{-n}c|\overline{b}_{-n}c) +(b_nc|b_nc)\\ 
			 &=
			(c|b_n\overline{b}_{-n}c)
			+ (c|\overline{b}_{-n}b_nc) 
			 =
			(c|[b_n,\overline{b}_{-n}]c) \\
			&=
			((L_0+1_V)^{-\frac{k}{2}}(L_0+1_V)^\frac{k}{2}c|[b_n,\overline{b}_{-n}]c)\\
			&=
			((L_0+1_V)^\frac{k}{2}c|[b_n,\overline{b}_{-n}](L_0+1_V)^{-\frac{k}{2}}c)\\
			 &\leq
			 \norm{(L_0+1_V)^\frac{k}{2}c}\,
			\norm{[b_n,\overline{b}_{-n}](L_0+1_V)^{-\frac{k}{2}}c} \\
			 &\leq
			M(\abs{n}+1)^s \norm{(L_0+1_V)^\frac{k}{2}c}^2 \,.
		\end{split}
	\end{equation}
Then the result follows taking the square roots.
\end{proof}

Let $C^\infty(S^1)$ be the vector space of infinitely differentiable complex-valued functions on $S^1$. Set the function $\chi(z):=e^{i\frac{x}{2}}$ for all $z\in S^1$ such that $z=e^{ix}$ for a unique $x\in(-\pi,\pi]$. Accordingly, we use $C^\infty_\chi(S^1)$ to denote the vector space of functions of kind $g=\chi h$ where $h\in C^\infty(S^1)$. With the following definitions of derivative:
\begin{equation}  \label{eq:defin_derivative_test_functions}
	\begin{split}
	f'(z)
		&:=
	\left.\frac{\mathrm{d}f(e^{ix})}{\mathrm{d}x}\right|_{e^{ix}=z}
	\\
	g'(z)
		&:=
	\left.\frac{\mathrm{d}g(e^{ix})}{\mathrm{d}x}\right|_{e^{ix}=z}=\chi(z)\left(\frac{i}{2}h(z)+h'(z)\right)
	\end{split}
\end{equation}
for all $z\in S^1$, we equip $C^\infty(S^1)$ and $C^\infty_\chi(S^1)$ with their natural Fréchet topologies. 
Therefore, for all non-negative $s\in\R$, the following family of norms are well-defined:
\begin{equation}
	\norm{f}_s:=\sum_{n\in\Z}(1+\abs{n})^s\abs{\widehat{f}_n}
	\quad\mbox{ and }\quad
	\norm{g}_s:=\sum_{n\in\Z-\half}(1+\abs{n})^s\abs{\widehat{g}_n}
\end{equation}
where we are using the Fourier coefficients
\begin{equation}
	\begin{split}
		\forall n\in\Z\qquad
		\widehat{f}_n
			&:=
		\oint_{S^1} f(z)z^{-n}\frac{\mathrm{d}z}{2\pi i z}
		=\frac{1}{2\pi}\int_{-\pi}^\pi f(e^{ix})e^{-inx} \mathrm{d}x
		 \\
		 \forall n\in\Z-\half\qquad
		\widehat{g}_n
			&:=
		\widehat{h}_\frac{2n-1}{2}
		=\frac{1}{2\pi}\int_{-\pi}^\pi h(e^{ix})e^{i\frac{x}{2}}e^{-inx} \mathrm{d}x
		 \,.
	\end{split}
\end{equation}

If $V$ is a unitary VOSA, call $\mathcal{H}$ its Hilbert space completion with respect to its scalar product $\scalar$. For all $a\in V$, the coefficients $a_n$ for all $n\in \half\Z$ are operators on $\mathcal{H}$ with dense domain $V$. Moreover, thanks to the invariance property of $\scalar$, every $a_n$ has densely defined adjoint and thus it is closable. We denote its closure with the same symbol.
Suppose further that $V$ is energy-bounded and let $a\in V_\parzero$ and $b\in V_\parone$, $f\in C^\infty(S^1)$ and $g\in C^\infty_\chi(S^1)$. Then the following operators
\begin{equation} \label{eq:defin_smeared_vertex_operators}
		Y_0(a,f)c:=	\sum_{n\in\Z}\widehat{f}_na_nc 
		\quad\mbox{ and }\quad
		Y_0(b,g)c:= \sum_{n\in\Z-\half}\widehat{g}_nb_nc 
\end{equation}
for all $c\in V$ are well-defined thanks to the energy bounds. Furthermore, they have densely defined adjoint in $\mathcal{H}$ thanks to the invariance property of the scalar product as just noted above. Therefore they are also closable: 

\begin{defin}
	The closures of the densely defined operators $Y_0(a,f)$ and $Y_0(b,g)$ defined in \eqref{eq:defin_smeared_vertex_operators} are called \textit{smeared vertex operators} and they are denoted by $Y(a,f)$ and $Y(b,g)$ respectively.
\end{defin}

It is proved, see \cite[Proposition 4.13]{CGH}, that the subspace $\mathcal{H}^\infty$ of $\mathcal{H}$ of \textit{smooth vectors} for $L_0$ is a common invariant core for the smeared vertex operators $Y(a,f)$ and $Y(b,g)$ and their adjoints. Note that $Y(\overline{a},\overline{f})\subseteq Y(a,f)^*$ and $Y(\overline{b},\overline{g})\subseteq Y(b,g)^*$, where $\overline{g}(z):=\chi(z)\overline{zh(z)}$ for all $z\in S^1$. It follows that if $a$ and $b$ are Hermitian, then the corresponding smeared vertex operators are self-adjoint whenever $f\in C^\infty(S^1,\R)$ and $g\in C^\infty_\chi(S^1,\R)$, that is whenever they are real-valued functions.

For a closed densely defined operator $A$ on a Hilbert space $\mathcal{K}$, the von Neumann algebra $W^*(A)$ generated by $A$ is given by:
\begin{equation}
	W^*(A)
	:=
	\{B\in B(\mathcal{K})\mid BA\subseteq AB \,,\,\, B^*A\subseteq AB^*\}' 
\end{equation}
where $\cdot '$ denotes the commutant in $B(\mathcal{K})$. 
If $\mathscr{F}$ is a family of closed densely defined operators on $\mathcal{K}$, the von Neumann algebra $W^*(\mathscr{F})$ generated by $\mathscr{F}$ is the smallest von Neumann algebra containing $\bigcup_{A\in\mathscr{F}}W^*(A)$.

Therefore, if $(V,\scalar)$ is a simple energy-bounded unitary VOSA, we define the isotonous family of von Neumann algebras on $\mathcal{H}$ generated from $(V,\scalar)$ by
\begin{equation}  \label{eq:defin_nets_from_VOSAs}
	\A_{(V,\scalar)}(I):=W^*\left(\left\{ Y(a,f) \,,\,\, Y(b,g) \,\left|\,
	\begin{array}{l}
		a\in V_\parzero \,,\,\, f\in C^\infty(S^1)\,,\,\,\mathrm{supp}f\subset I \\
		b\in V_\parone \,,\,\, g\in C^\infty_\chi(S^1) \,,\,\,\mathrm{supp}g\subset I
	\end{array} 
	\right.\right\}\right)
\end{equation}
for all $I\in\J$. In \cite[Section 4.2]{CGH}, it is proved that the vacuum vector $\Omega$ of $V$ is a cyclic vector for the von Neumann algebra $A_{(V,\scalar)}(S^1)$. 
Moreover, the \textit{positive-energy unitary representation} of the Virasoro algebra on $V$, arising from the coefficients of $Y(\nu, z)$, integrates to a positive-energy strongly continuous projective unitary representation $U$ of $\Diff^+(S^1)^{(\infty)}$ such that
\begin{equation}
	U(\exp^{(\infty)}(tf))AU(\exp^{(\infty)}(tf))^*
		=
	e^{itY(\nu,f)}Ae^{-itY(\nu,f)}
\end{equation}
where $\exp^{(\infty)}(tf)$ with $t\in\R$ is the one-parameter subgroup of $\Diff^+(S^1)^{(\infty)}$ generated by the smooth real vector field $f$ on $S^1$. 
Note that in this paper, we identify a smooth vector field $f\frac{\mathrm{d}}{\mathrm{d}x}$ with its corresponding smooth function $f$.
In particular, $U$ factors through a representation of the double cover $\Diff^+(S^1)^{(2)}$ as $U(r^{(\infty)}(2\pi))=\Gamma$, which is the extension to $\mathcal{H}$ of $\Gamma_V$. 
Therefore, it is proved that the family \eqref{eq:defin_nets_from_VOSAs} satisfies M\"{o}bius covariance with respect to $U$. It is also irreducible as the vacuum $\Omega$ is unique, $V$ being of CFT type.

Unfortunately, the graded locality of vertex operators is not in general enough to conclude the graded locality of the family \eqref{eq:defin_nets_from_VOSAs} with respect to $\Gamma$ and $Z$, that is the extension to $\mathcal{H}$ of the vector space map $Z_V:=\frac{1-i\Gamma_V}{1-i}$ on $V$. 
Note that $Z^*=Z^{-1}$ and that $Z(a)=(-i)^pa$ for all $a\in V_{\overline{p}}$.
Then we come to the following:

\begin{defin}
	A unitary VOSA $V$ is said to be \textit{strongly graded-local} if it is energy-bounded and $\A_{(V,\scalar)}$ satisfies the graded locality. 
\end{defin}

Therefore, if $(V,\scalar)$ is a simple strongly graded-local unitary VOSA, it can be proved that $\A_{(V,\scalar)}$ is also diffeomorphism covariant with respect to $U$ and thus it defines a proper irreducible graded-local conformal net on $\mathcal{H}$, which is also independent, up to isomorphism, of the choice of the scalar product on $V$. Accordingly, we denote such net simply by $\A_V$. 

Let $V^1$ and $V^2$ be two VOSAs. Then we denote by $V:=V^1\hat{\otimes}V^2$ their \textit{graded tensor product}, see \cite[Section 3.1]{CGH} and references therein. We highlight that if those VOSAs are also unitary, then the unitary structure $(\scalar,\theta)$ on $V$ is given by $\scalar:=\scalar_1\scalar_2$ and by $\theta:=\theta_1\otimes \theta_2$, where the indexes denote the unitary structures on the corresponding VOSAs, see \cite[Proposition 2.4]{AL17} and \cite[Proposition 2.20]{Ten19}. By \cite[Corollary 6.6]{CGH}, $V$ is strongly graded-local if and only if $V^1$ and $V^2$ are. In this case, we also have that $\A_V$ is isomorphic to the graded tensor product $\A_{V^1}\hat{\otimes}\A_{V^2}$, see \cite[Example 2.7]{CGH} and references therein.

\begin{ex} \label{ex:real_free_fermion}
	The easiest example of the correspondence between simple unitary VOSAs and irreducible graded-local conformal nets is given by the \textit{real free fermion models}, see \cite[Example 7.1]{CGH} and references therein. We use $F$ for the \textit{real free fermion VOSA}. Recall that $F$ is a simple unitary VOSA with central charge $\half$, which can be generated by any Hermitian primary vector $\fer\in F_\half$ with norm one. Moreover, $F$ is strongly graded-local and $\F:=\A_F$ is called the \textit{real free fermion net}. The graded tensor product $F^n$ of $n$ copies of $F$ is strongly graded-local and we denote the corresponding graded-local conformal net by $\F^n$. In particular, $F^2$ and $\F^2$ are known as the \textit{charged free fermion VOSA} and \textit{net} respectively.
\end{ex}

	Let $V$ be a vertex superalgebra and let $\nu$ be any Virasoro vector with $Y(\nu, z)=\sum_{m\in\Z}L_nz^{-n-2}$. We call $\tau\in V_\parone$ a \textit{super-Virasoro vector} of $V$ (with respect to $\nu$) if the coefficients of $Y(\tau,z)=\sum_{n\in\Z-\half}G_nz^{-n-\frac{3}{2}}$ satisfy with the ones of $Y(\nu,z)$ the commutation relations of the \textit{Neveu-Schwarz algebra} $NS$ for some central charge $c\in\C$:
	\begin{equation}  \label{eq:NS_cr}
		\begin{split}
			\forall m,n\in\Z \qquad
			[L_m,L_n]  &:=
			(m-n)L_{m+n}+\frac{c(m^3-m)}{12}\delta_{m,-n}
			\\
			\forall m\in\Z
			\,\,\,
			\forall n\in\Z-\half \qquad
			[L_m,G_n] &:= 
			\left(\frac{m}{2}-n\right) G_{m+n}
			\\
			\forall m,n\in\Z-\half \qquad
			[G_m,G_n] &:=
			2L_{m+n}+\frac{c}{3}\left(m^2-\frac{1}{4}\right)\delta_{m,-n}
			\,.
		\end{split}
	\end{equation} 
	Suppose that $\nu$ makes $V$ into a conformal vertex superalgebra.
	This means that the Virasoro vector $\nu$ is a conformal vector, that is $L_{-1}=T$ and $L_0$ is diagonalizable. In this case, $\tau$ is called a \textit{superconformal vector}, see \cite[Definition 5.9 and Proposition 5.9]{Kac01}.
	Then a VOSA $V$ is a \textit{$N=1$ superconformal VOSA} if there is a superconformal vector $\tau$ associated to the conformal vector $\nu$ of $V$. If $V$ is unitary, then it is called a \textit{unitary $N=1$ superconformal VOSA} if the vertex subalgebra generated by $\nu$ and $\tau$ is a unitary subalgebra, see \cite[Definition 7.7]{CGH}, which turns out to be equivalent to ask for $\tau$ to be Hermitian if $V$ is also simple, see \cite[Theorem 7.9(i)]{CGH}. Therefore, if $V$ is a simple $N=1$ superconformal VOSA which is also strongly graded-local, then $V$ is unitary $N=1$ superconformal if and only if $\A_V$ is an irreducible \textit{$N=1$ superconformal net} in the sense of \cite[Definition 7.6]{CGH}, see \cite[Theorem 7.9(ii)]{CGH}. 

Now, we give a well-known criteria useful to prove the strong graded locality of a unitary VOSA, based on linear energy bounds.
Set $C^\infty_\parzero(S^1):=C^\infty(S^1)$ and $C^\infty_\parone(S^1):=C^\infty_\chi(S^1)$ and corresponding symbols for their real-valued subsets. 

\begin{theo}  \label{theo:linear_energy_bounds_strong_locality}
	Let $V$ be a simple energy-bounded unitary VOSA. If $a$ is a Hermitian quasi-primary vector with given parity $p(a)$ satisfying linear energy bounds, then $W^*(Y(a,f))\subseteq W^*(ZY(b,g)Z^*)'$ for all $b\in V$ with given parity $p(b)$, all $f\in C^\infty_{p(a)}(S^1,\R)$ and all $g\in C^\infty_{p(b)}(S^1)$ whenever $f$ and $g$ have disjoint supports. 
	This is true in particular if either $a$ is any vector in $V_\half\cup V_1$ or $a$ is a Hermitian quasi-primary Virasoro or super-Virasoro vector.
\end{theo}

\begin{proof}
	The first part of this theorem can be proved as in the last part of the proof of \cite[Lemma 3.6]{CTW22}, see also \cite{CTW23b} and the proof of \cite[Theorem 7.15]{CGH}. The second part follows recalling that: any vector in $V_\half$ satisfies $0$-th order energy bounds, see \cite[Proposition 4.5]{CGH}; vectors in $V_1$ satisfies $\half$-th order energy bounds, see \cite[Proposition 3.6]{CT23}; Hermitian quasi-primary Virasoro vectors satisfy linear energy bounds, see \cite[Proposition 3.4]{CT23}; Hermitian quasi-primary super-Virasoro vectors satisfy $\half$-th order energy bounds, see \cite[Eq.\ (27)]{CKL08}.
\end{proof}

\section{Strong graded locality}
\label{section:strong_graded_locality}

Let $\g=\g_\parzero\oplus \g_\parone$ be a simple finite-dimensional Lie superalgebra, see \cite{Kac77}, with a non-degenerate even supersymmetric invariant bilinear form $B\bilinear$. Let $x$ and $f$ be two even elements in $\g_\parzero$ such that $\mathrm{ad}x$ is diagonalizable on $\g$ with half-integer eigenvalues and that $[x,f]=-f$. Note that these conditions imply that $f$ is a nilpotent element of $\g$. We denote the eigenspaces of $\mathrm{ad}x$ by $\g_j$ with $j\in\half\Z$. 
Then the pair $(x,f)$ is called \textit{good} if the centralizer $\g^f=\{a\in\g\mid [f,a]=0\}$ of $f$ in $\g$ decomposes as
\begin{equation}
	\g^f=\bigoplus_{j\in\half\Z_{\leq 0}} \g^f_j
	\qquad\mbox{where}\qquad
	\g^f_j=\{a\in \g_j\mid [f,a]=0\} \,.
\end{equation}
In this case, $\mathrm{ad}x$ realizes an equivalence between $\g_\half$ and $\g_{-\half}$, see e.g.\ \cite[Eq.s (1.10)--(1.11)]{KW04}.
A pair $(x,f)$ is automatically good if there is an element $e\in\g_\parzero$ such that $[x,e]=e$, that is $(x,f)$ is part of a \textit{$\mathfrak{sl}_2$-triple} $(e,x,f)$ with $e\in \g_\parzero$. In this case, $e$ is uniquely determined by $(x,f)$. Moreover, $(\g, x, f)$ is called a \textit{Dynkin datum}.
Now, suppose that $\g$ is also \textit{basic}, which means that $\g_\parzero$ is reductive. Then it is not difficult to see that $e$, $x$ and $f$ belong to the semisimple part of $\g_\parzero$, so that $x$ is uniquely determined, up to conjugation, by the nilpotent element $f$, see \cite[Corollary 3.7 ]{Kos59}.
By a procedure of \textit{quantum (Hamiltonian) reduction}, see \cite[Section 2]{KRW03} and \cite[Sections 1--4]{KW04}, see also \cite[Section 5.1]{DK06}, it is possible to associate to every Dynkin datum $(\g,x,f)$, a family of vertex algebras, in the sense of \cite{Kac01}, $W^k(\g,x, f)$ with $k\in\C$, called \textit{universal $W$-algebras}. 
Moreover, $W^k(\g,x,f)$ has a natural conformal vector $\nu$, making it into a conformal vertex superalgebra, whenever $k\not=-h^\vee$, where $h^\vee$ is the \textit{dual Coxeter number} of $\g$. In this case, we denote the simple quotient of $W^k(\g,x,f)$ by $W_k(\g,x,f)$. 
It turns out that $W^k(\g,x, f)$ and its simple quotient $W_k(\g,x,f)$ depend, up to isomorphisms, only on $\g$, $f$ and $k$. Accordingly, as usual, we will use the symbols $W^k(\g,f)$ and $W_k(\g,f)$ to denote $W^k(\g,x, f)$ and $W_k(\g,x,f)$ respectively for a given choice of $x$. 

\textit{Minimal $W$-algebras}, introduced in \cite[Section 4]{KRW03} and \cite[Section 5]{KW04}, see also \cite[Section 2]{AKMPP18}, \cite[Section 7]{KMP22} and \cite[Section 7]{KMP23}, arise for a particular class of Dynkin datum, called \textit{minimal}, characterized by the fact that $\mathrm{ad}x$ gives a \textit{minimal gradation} of $\g$, that is
\begin{equation}  \label{eq:minimal_gradation_g}
	\g=\g_{-1}\oplus \g_{-\half}\oplus \g_0 \oplus \g_\half\oplus \g_1
	\,,\qquad \g_{-1}=\C f \,,\,\,\, \g_1=\C e \,.
\end{equation}
A key feature of universal minimal $W$-algebras is that they have explicit sets of free strong generators, which we are going to describe in the following, after setting some notations. 
Let $\g^\natural$ be the centralizer of $\{e, x, f\}$ in $\g$, then by the property of $B\bilinear$, one can see that
$$
\g_0=\C x\oplus \g^\natural\,,\qquad
\g^\natural=\{a\in \g_0\mid B(a,x)=0\} \,.
$$
It follows that if $\h$ is a Cartan subalgebra of the even part of $\g_0$, then $\h=\C x\oplus\h^\natural$, where $\h^\natural:=\{a\in \h\mid B(a,x)=0\}$ is a Cartan subalgebra of the even part of $\g^\natural$. Moreover, $\h$ is a Cartan subalgebra of $\g_\parzero$ too.
We also normalize the bilinear form $B\bilinear$ by the condition $B(x,x)=\half$. This also fixes the Casimir operator of $\g$ and its eigenvalue $2h^\vee$ on $\g$. Therefore, a complete list of minimal gradation for $\g$, along with a description of $\g^\natural$, $\g_\half$ and $h^\vee$, is given in \cite[Tables 1--3]{KW04}. 

\begin{rem}
	Let $\Delta\subset \h^*$ be the set of roots of $\h$ in $\g$.
	Then, a minimal Dynkin datum $(\g,x,f)$ can be obtained choosing $f$ as the root vector of $\g$ attached to a minimal root, so that $e$ is the root vector attached to the corresponding maximal root. Actually, it can be showed that there is a bijection between minimal gradations of $\g$, up to automorphisms of $\g$, and highest roots of the simple components of $\g_\parzero$, which can be made highest roots of $\g$ for some ordering of $\Delta$, up to the action of the Weyl group, see \cite[Section 5]{KW04}. 
\end{rem}

By \cite[Theorem 5.1]{KW04} with \cite{KW05}, cf.\ also \cite[Section 5.3]{DK06}, 
there are two linear isomorphisms $\g^\natural \ni u\mapsto J^{\{u\}}\in W^k(\g,f)_1$ and $\g_{-\half}\ni v\mapsto G^{\{v\}}\in W^k(\g,f)_\frac{3}{2}$ such that the images of the elements in any pair of bases of $\g^\natural$ and of $\g_{-\half}$ form together with the conformal vector $\nu$ a set of free strong generators for $W^k(\g,f)$.   
The $\lambda$-brackets are as follows:  
\begin{equation} \label{eq:lambda-bracket_generators}
	\begin{split}
		[{J^{\{u\}}}_\lambda G^{\{v\}}]
			&=G^{\{[u,v]\}} \\
		[{J^{\{u\}}}_\lambda J^{\{v\}}]
			&= 
		J^{\{[u,v]\}}+\lambda\left[\left(k+\frac{h^\vee}{2}\right)B(u,v)-\frac{1}{4}\kappa_{\g_0}(u|v)\right]
	\end{split}
\end{equation}
where $\kappa_{\g_0}\scalar$ is the Killing form of $\g_0$ and (see \cite[Eq.\ (7.7)]{KMP23}, which is derived from \cite[Theorem 5.1(e)]{KW04}, but with a correct version of the $\lambda^2$-term, recall \cite{KW05}, from \cite[Eq.\ (3.1)]{AKMPP18})
\begin{equation} \label{eq:lambda-bracket_G_generators}
	\begin{split}
		[{G^{\{u\}}}_\lambda G^{\{v\}}]
			&=
		-2(k+h^\vee)\langle u,v\rangle \nu
		+\langle u,v\rangle 	\sum_{\alpha=1}^{\dim\g^\natural}:J^{\{u^\alpha\}} J^{\{u_\alpha\}}: \\
			&+2\sum_{\alpha,\beta=1}^{\dim\g^\natural}\langle [u_\alpha,u], [v,u^\beta] \rangle :J^{\{u^\alpha\}}J^{\{u_\beta\}}: 
		+2(k+1)L_{-1}J^{\{[[e,u],v]^\natural\}} \\
			&+2\lambda\left(J^{\{[[e,u],v]^\natural\}}+\sum_{\alpha,\beta=1}^{\dim\g^\natural}\langle [u_\alpha,u], [v,u^\beta] \rangle J^{\{[u^\alpha,u_\beta]\}}\right) \\
			&+2\lambda^2\langle u,v\rangle p(k)\Omega 
	\end{split}
\end{equation}
where: $\langle u,v\rangle:=B(f,[u,v])$ for all $u,v\in\g_{-\half}$; $\{u_\alpha\}$ and $\{u^\alpha\}$ are dual basis of $\g^\natural$ with respect to $B\bilinear$; $a\mapsto a^\natural$ is the orthogonal projection from $\g_0$ to $\g^\natural$; $p(k)$ is a monic quadratic polynomial from \cite[Table 4]{AKMPP18}.
Note that these generators descend to strong generators for the simple minimal $W$-algebras.
The central charge is given by \cite[Eq.\ (5.7)]{KW04}:
\begin{equation}  \label{eq:central_charge}
	c(k)=\frac{kd}{k+h^\vee}-6k+h^\vee-4 \,,\quad d:=\operatorname{sdim}(\g):=\dim(\g_\parzero)-\dim(\g_\parone) \,.
\end{equation}

\begin{defin}[{\cite[Definition 3.1]{AKMPP18}}]
	Let $W^k(\g,f)$ be a minimal $W$-algebra with its simple quotient $W_k(\g,f)$, where $k$ is any complex number different from $-h^\vee$. Set $\mathcal{V}^k(\g^\natural)$ as the vertex subalgebra of $W^k(\g,f)$ generated by the elements $J^{\{u\}}$ with $u\in\g^\natural$. Accordingly, set $\mathcal{V}_k(\g^\natural)$ as the image of $\mathcal{V}^k(\g^\natural)$ in $W_k(\g,f)$. If there exists $k\in\C$ such that $W_k(\g,f)$ is equal to $\mathcal{V}_k(\g^\natural)$, then $k$ is called a \textit{collapsing level} and one says that $W_k(\g,f)$ \textit{collapses} to $\mathcal{V}_k(\g^\natural)$.
\end{defin}

Now, we move to the unitarity of minimal $W$-algebras. 
First of all, from \cite[Theorem 3.3 and Proposition 3.4]{AKMPP18}, we have that if $k$ is a collapsing level, then the minimal $W$-algebra $W_k(\g,f)$ is either the \textit{trivial VOA} $\C\Omega$ or it collapses to a simple \textit{affine vertex superalgebra}, see e.g.\ \cite[Section 4.7 and Section 5.7]{Kac01}. In the former case, $W_k(\g,f)$ is trivially unitary, whereas in the latter one it is unitary if and only if the corresponding affine vertex superalgebra is a unitary VOSA; more details about these cases will be given later in Remark \ref{rem:unitary_range_and_collapsing}. 
Therefore, we can restrict to the non-collapsing levels.
In this case, if $k$ has non-zero imaginary part, then it follows from \eqref{eq:lambda-bracket_generators} that $W_k(\g,f)$ cannot be unitary.
If $k$ is real, by \cite[Theorem 7.9]{KMP22}, $W_k(\g,f)$ can be a unitary VOSA only if the $\mathrm{ad}x$-gradation \eqref{eq:minimal_gradation_g} is compatible with the parity of $\g$, that is the parity of $\g_j$ is $2j$ mod $2$ for $j\in\{\pm\half,0,\pm 1\}$. 
Under this compatibility condition, for all $k\in \C\backslash\{-h^\vee\}$, $W^k(\g,f)$ has a structure of VOSA, so that $W_k(\g,f)$ is a simple (possibly non-unitary) VOSA.
A complete list of possible choices for $\g$ satisfying this compatibility condition is given by the entries of \cite[Table 2]{KW04} together with the case $\g=\mathfrak{sl}_2$.
One can note from this list that the $\mathrm{ad}x$-gradation \eqref{eq:minimal_gradation_g} is unique up to automorphisms of $\g$, so that different choices for $f$ will produce isomorphic VOSAs, see also Remark \ref{rem:lie_superalgebras_iso}.
Among possible choices for $\g$, there are three special cases which are the only ones with $\g^\natural$ abelian.
These are $W_k(\mathfrak{sl}_2,f)$, $W_k(\mathfrak{spo}(2|1),f)$ and $W_k(\mathfrak{spo}(2|2),f)$, giving rise to the well-known \textit{Virasoro VOAs} $L(c(k))$, \textit{$N=1$ super-Virasoro VOSAs} $V^{c(k)}(NS)$ and \textit{$N=2$ super-Virasoro VOSAs} $V^{c(k)}(N2)$ respectively with corresponding central charges as in \eqref{eq:central_charge}, see e.g.\ \cite[Remark 4.1]{KRW03} and \cite[Section 8]{KW04}.
Their \textit{unitary series} were already established in the `80s, see Sections \ref{subsec:virasoro_nets}--\ref{subsec:N=2_super-virasoro_nets} respectively and references therein.
The remaining non-trivial unitary minimal $W$-algebras $W_k(\g,f)$ with $\g^\natural$ non-abelian are classified by \cite[Corollary 11.2 and Proposition 8.10]{KMP23}, cf.\ also \cite[Table 1]{KMP25}. 
Following \cite{KMP23}, for any $\g$ from \cite[Table 2]{KW04} with $\g^\natural$ non-abelian, we call the \textit{unitary range} of $W_k(\g,f)$ the set of $k\in\R$ different from $-h^\vee$ and such that $W_k(\g,f)$ is non-trivial and unitary. 

\begin{rem} \label{rem:lie_superalgebras_iso}
	When inspecting \cite[Table 2]{KW04}, recall that: for all $m\in\Zplus$ and all even $n\in\Z_{\geq 2}$, $\mathfrak{osp}(m|n)\cong\mathfrak{spo}(n|m)$; $\mathfrak{spo}(2|2)\cong \mathfrak{sl}(2|1)$; $\mathfrak{spo}(2|4)\cong D(2,1;1)$; for all $m\in\Z_{\geq 2}$, $\mathfrak{gl}_m\cong \C\oplus \mathfrak{sl}_m$; $\mathfrak{so}_4\cong \mathfrak{sl}_2\oplus \mathfrak{sl}_2$; $\mathfrak{spo}(2|0)\cong \mathfrak{sl}_2\cong \mathfrak{sp}_2$.
	Moreover, $D(2,1;a)$ and $D(2,1;a')$ are isomorphic if and only if $a$ and $a'$ lie in the same orbit of the group generated by the transformations $b\mapsto b^{-1}$ and $b\mapsto -b-1$.
\end{rem}

\begin{rem} \label{rem:unitary_range_and_collapsing}
For the readers' convenience, we collect here below some useful data about the unitary ranges, also specifying those collapsing levels giving rise to unitary VOAs:
	\begin{itemize}
		\item $\g=\mathfrak{sl}(2|m)$ for all $m\geq 3$, $k=-1$, $h^\vee=2-m$, $d=m^2-4m+3$, $W_{-1}(\g,f)\cong M(1)$;
		
		\item $\g=\mathfrak{psl}(2|2)$, $k\in \Z_{\leq -2}$, $h^\vee=0$, $d=-2$, $W_{-1}(\g,f)=\C\Omega$;
		
		\item $\g=\mathfrak{spo}(2|3)$, $k\in \frac{1}{4}\Z_{\leq -3}$, $h^\vee=\half$, $d=0$, $W_{-\half}(\g,f)=\C\Omega$, $W_{-\frac{3}{4}}(\g,f)\cong V_1(\mathfrak{sl}_2)$ with $c(-\frac{3}{4})=1$;
		
		\item $\g=\mathfrak{spo}(2|m)$ for all $m\geq 4$, $k\in \half\Z_{\leq -2}$, $h^\vee=2-\frac{m}{2}$, $d=3+\frac{m(m-5)}{2}$, $W_{-\half}(\g,f)=\C\Omega$;
		
		\item $\g=D(2,1;\frac{m}{n})$ for all $m,n\in\Zplus$ coprime and $(m,n)\not=(1,1)$, $k\in \frac{mn}{m+n}\Z_{<0}$, $h^\vee=0$, $d=1$, $W_{\frac{-m}{m+1}}(\g,f)\cong V_{m-1}(\mathfrak{sl}_2)$ with $c(\frac{-m}{m+1})=\frac{3(m-1)}{m+1}$ for all $m\in \Z_{\geq 2}$;
		
		\item $\g=F(4)$, $k\in \frac{2}{3}\Z_{\leq -2}$, $h^\vee=-2$, $d=8$, $W_{-\frac{2}{3}}(\g,f)=\C\Omega$;
		
		\item $\g=G(3)$, $k\in \frac{3}{4}\Z_{\leq -2}$, $h^\vee=-\frac{3}{2}$, $d=3$, $W_{-\frac{3}{4}}(\g,f)=\C\Omega$;
	\end{itemize}
where $M(1)$ is the \textit{Heisenberg VOA} with central charge $c=1$, see e.g.\ \cite[Section 3.5 and Proposition 4.10(a)]{Kac01} and \cite[Section 4.3]{DL14}, and $V_1(\mathfrak{sl}_2)$ and $V_{m-1}(\mathfrak{sl}_2)$ are the \textit{unitary} affine VOAs \cite[Section 4.2]{DL14} associated to the simple complex Lie algebra $\mathfrak{sl}_2$ at levels $1$ and $m-1$ respectively. 
As we have explained above, outside the unitary ranges together with the Virasoro, $N=1$ and $N=2$ super-Virasoro cases, see \cite[Table 1 and Table 2]{KW04}, the unitary minimal $W$-algebras are all collapsing. The ones collapsing to the trivial VOA are listed in \cite[Proposition 3.4]{AKMPP18}, whereas the non-trivial ones are given by \cite[Proposition 7.11 and Corollary 7.12]{KMP22}, cf.\ \cite[Theorem 7.4]{KMP23}:
\begin{itemize}
	\item $W_{-\frac{4}{3}}(G_2,f)\cong V_1(\mathfrak{sl}_2)$ with $c(-\frac{4}{3})=1$;
	
	\item $W_{-1}(\mathfrak{sl}(m|n),f)\cong M(1)$ for all $n\geq 1$ and all $m>2$ such that $m\not\in\{n, n+1, n+2\}$;
	
	\item $W_{-2}(\mathfrak{osp}(m|n),f)\cong V_\frac{{m-n-8}}{2}(\mathfrak{sl}_2)$ with $c(-2)=\frac{3(m-n-8)}{m-n-4}$ for all $m,n\geq 1$ such that $m-n\geq 10$ and even.
\end{itemize}
\end{rem}

\begin{rem}  \label{rem:unitary_structure_minimal_W-algebras}
	We point out some facts about the unitary structure of minimal $W$-algebras. 
	This is determined by a suitable involutive antilinear automorphism $\theta$ of the Lie superalgebra $\g$. In particular, $\theta$ must fix the elements $e$, $x$ and $f$ of the $\mathfrak{sl}_2$-triple. This implies that $\theta$ preserves the eigenspaces of $\mathrm{ad}x$, that is $\theta(\g_j)=\g_j$ for all $j\in\{\pm\half,0,\pm 1\}$. 
	Then, see \cite[Proposition 7.1 and Proposition 7.2]{KMP23}, the above involution $\theta$ determines the PCT operator, which we still denote by $\theta$, of a unitary minimal $W$-algebra by $\theta(J^{\{u\}})=J^{\{\theta(u)\}}$ and $\theta(G^{\{v\}})=G^{\{\theta(v)\}}$.
	It follows that we can choose a pair of bases of $\g^\natural$ and of $\g_{-\half}$ in such a way that the corresponding generators are Hermitian, see also the discussion after the proof of \cite[Lemma 7.3]{KMP22}, that is $\theta(J^{\{u\}})=-J^{\{u\}}$ and $\theta(G^{\{v\}})=G^{\{v\}}$.
\end{rem}

We are now ready to prove the main theorem of this section about the strong graded locality of unitary minimal $W$-algebras.

\begin{theo} \label{theo:W-algebras_nets}
	Every unitary minimal W-algebra $W_k(\g,f)$ is strongly graded-local. Then to every such W-algebra is associated a unique, up to isomorphism, irreducible graded-local conformal net $\A_{W_k(\g,f)}$.
\end{theo}

\begin{proof}
First of all, note that if $W_k(\g,f)=\C\Omega$, then it is trivially strongly local. Moreover, the Heisenberg VOA $M(1)$ and the unitary affine VOAs are already known to be strongly local, see \cite[Exmaple 8.6 and Example 8.7]{CKLW18} respectively. 
For the cases with $\g^\natural$ abelian, the strong graded locality of $W_k(\g,f)$ is also already known, see Sections \ref{subsec:virasoro_nets}--\ref{subsec:N=2_super-virasoro_nets} and references therein.  

Let $W_k(\g,f)$ be any unitary minimal $W$-algebra such that $k$ is not a collapsing level and $\g^\natural$ is not abelian. 
As explained in Remark \ref{rem:unitary_structure_minimal_W-algebras}, we can consider a set of Hermitian generators $\nu$, $J^{\{u\}}$ and $G^{\{v\}}$.
We are going to prove that these generators, with the $\lambda$-brackets in \eqref{eq:lambda-bracket_generators}--\eqref{eq:lambda-bracket_G_generators}, satisfy linear energy bounds. It is well-known that the conformal vector $\nu$ satisfies linear energy bounds and the currents $J^{\{u\}}$ satisfy $\half$-th order energy bounds, see the proof of Theorem \ref{theo:linear_energy_bounds_strong_locality}. Then it remains to prove it for the primary generators of type $G^{\{v\}}$.

First, we have that $\overline{G^{\{v\}}}=\theta(G^{\{v\}})=G^{\{v\}}$. Recalling the formula for the $\lambda$-bracket \eqref{eq:lambda_bracket}, we can calculate the vectors $G^{\{v\}}_{(j)}G^{\{v\}}$ appearing in the commutators $[G^{\{v\}}_n, G^{\{v\}}_{-n}]$ for all $n\in\Z-\half$, according to the Borcherds commutator formula \eqref{eq:borcherds_commutator_formula_weighted}. Indeed, from \eqref{eq:lambda-bracket_G_generators} we have that $G^{\{v\}}_{(j)}G^{\{v\}}=0$ for all $j>2$ and
\begin{equation}
	\begin{split}
		G^{\{v\}}_{(0)}G^{\{v\}}
			&= 
		-2(k+h^\vee)\langle v,v\rangle \nu
		+\langle v,v\rangle 	\sum_{\alpha=1}^{\dim\g^\natural}:J^{\{u^\alpha\}} J^{\{u_\alpha\}}: \\
			&
		+2\sum_{\alpha,\beta=1}^{\dim\g^\natural}\langle [u_\alpha,v], [v,u^\beta] \rangle :J^{\{u^\alpha\}}J^{\{u_\beta\}}: 
		+2(k+1)L_{-1}J^{\{[[e,v],v]^\natural\}} \\
		G^{\{v\}}_{(1)}G^{\{v\}}
			&=
		2J^{\{[[e,v],v]^\natural\}}+2\sum_{\alpha,\beta=1}^{\dim\g^\natural}\langle [u_\alpha,v], [v,u^\beta] \rangle J^{\{[u^\alpha,u_\beta]\}} \\
		G^{\{v\}}_{(2)}G^{\{v\}}
			&=
		4\langle v,v \rangle p(k)\Omega	\,.
	\end{split}
\end{equation}
It is then manifest that $G^{\{v\}}_{(1)}G^{\{v\}}$ and $G^{\{v\}}_{(2)}G^{\{v\}}$ satisfy $\half$-th order energy bounds as currents do it. By \cite[Eq.\ (102)]{CKLW18}, for any basis element $u\in\g^\natural$, $L_{-1}J^{\{u\}}$ satisfies $\frac{3}{2}$-th order energy bounds; whereas for basis elements $u,v\in\g^\natural$, $:J^{\{u\}}J^{\{v\}}:$ satisfies $2$-nd order energy bounds, see \cite[Eq.s (102) and (104)]{CKLW18} with the discussion thereafter. It follows that $G^{\{v\}}_{(0)}G^{\{v\}}$ satisfies $2$-nd order energy bounds. To sum up, for all $n\in\Z-\half$, the commutator $[G^{\{v\}}_n, G^{\{v\}}_{-n}]$ satisfies the estimate in the hypotheses of Proposition \ref{prop:energy_bounds_odd_vectors_anticommutator} with $k=2$ there. Then for all basis elements $v\in \g_{-\half}$, $G^{\{v\}}$ satisfies linear energy bounds as desired. 

By Theorem \ref{theo:linear_energy_bounds_strong_locality}, $W^*(Y(a,f))\subseteq W^*(ZY(b,g)Z^*)'$ with $b\in W_k(\g,f)$ whenever $a$ is any of the generators of $W_k(\g,f)$, $f$ and $g$ are suitable test functions with disjoint supports. Therefore, $W_k(\g,f)$ is strongly graded-local by \cite[Theorem 6.4]{CGH}.
\end{proof}

\section{Representation theory}
\label{section:representations}

In this section, we investigate the representation theories of some of the \textit{minimal $W$-algebra nets} $\A_{W_k(\g,f)}$ arising from Theorem \ref{theo:W-algebras_nets}.
After giving some definitions with related bibliographical references, we present some results on the representation theory of graded-local conformal nets.

The definitions of \textit{weak}, \textit{admissible} and \textit{ordinary} $V$-\textit{module} for a VOSA $V$ can be found in e.g.\ \cite[Section 4]{DZ05}, \cite[Section 2]{DZ06}, see also \cite{ABD04}. If not specified, by a $V$-module, we mean an ordinary $V$-module.
A VOSA $V$ is said to be \textit{self-contragredient} if, as $V$-module, it is isomorphic to its contragredient module. This property is equivalent to the existence of a non-degenerate \textit{invariant bilinear form} on $V$, see e.g.\ \cite[Definition 3.6]{CGH}, so that self-contragredient VOSAs of CFT type are automatically simple by \cite[Proposition 3.8 (iv)]{CGH}.
A VOSA $V$ is said to be \textit{$C_2$-cofinite} if  
$V/\langle v_{(-2)}u\mid v,u\in V\rangle$ is finite-dimensional, where $\langle \cdot \rangle$ denotes the linear span. 
It is said to be \textit{rational} if every admissible $V$-module is a direct sum of irreducible ordinary $V$-modules.
By \cite[Theorem 6.6]{DZ06}, the rationality implies that $V$ has only a finite number of irreducible $V$-modules.

\begin{defin}
	A VOSA $V$ is \textit{strongly rational} if it is of CFT type, self-contragredient, $C_2$-cofinite and rational (and thus also simple). 
\end{defin}

\begin{rem}
	By \cite[Proposition 2.4]{HM}, if $V$ is a strongly rational VOSA, then
	$V_\parzero$ is strongly rational too. The converse is also true if $V_\parone$ satisfies certain conditions as $V_\parzero$-module as we explain in the following. If $V$ is a simple VOSA, then $V_\parzero$ is simple and $V_\parone$ is an irreducible $V_\parzero$-module, see e.g.\ \cite[Remark 4.17]{CGGH23}. If it is also of CFT type, then $V_\parzero$ is of CFT type and $V_\parone$ has positive $L_0$-grading. Moreover, if $V_\parzero$ is strongly rational, then its modules form a modular tensor category, see e.g.\ \cite{Hua08} and references therein. Therefore, $V_\parone$ turns out to be a simple current of $V_\parzero$, see \cite[Theorem 3.1 and Remark A.2]{CKLR19} and the hypotheses of \cite[Proposition 2.6]{HM} apply, so that $V$ is strongly rational too.
	To sum up, we have that a simple VOSA $V$ of CFT type (which will be always our case) is strongly rational if and only if $V_\parzero$ is.
\end{rem}

Let $\A$ be an irreducible graded-local conformal net. 
Then $\A$ satisfies the \textit{twisted Haag duality}, see \cite[Theorem 5]{CKL08}, that is for all $I\in\J$,
\begin{equation}  \label{eq:twisted_Haag_duality}
	\A(I')=Z\A(I)'Z^*=Z^*\A(I)'Z 
	\,.
\end{equation}
Let $I\in\J$, consider any point $p\in I$ and call $I_1, I_2\in \J$ the two connected components of $I\backslash\{p\}$. Then we say that $\A$ is \textit{strongly additive} if $\A(I)$ is equal to the von Neumann algebra $\A(I_1)\vee\A(I_2)$ generated by $\A(I_1)$ and $\A(I_2)$.
Furthermore, we say that $\A$ satisfies the \textit{split property} if for all $I_0,I\in\J$ such that $\overline{I_0}\subset I$, there exists a type I factor $\mathcal{R}$ such that $\A(I_0)\subset \mathcal{R} \subset\A(I)$.
A classical argument, see e.g.\ \cite[Lemma 5.4.2]{BW92}, shows that the split property is equivalent to require that $\A(I_0)\otimes \A(I)'$ is naturally spatially isomorphic to $\A(I_0)\vee \A(I)'$, that is the map  $AB \mapsto A \otimes B$ with $A \in \A(I_0)$ and $B\in \A(I)'$, extends to a spatial isomorphism. See also \cite{DL83} for other characterizations of the split property. 
Thanks to twisted Haag duality, we also have that $\A$ is split if and only if $\A(I_0)\otimes Z\A(I')Z^*$ is naturally spatially isomorphic to $\A(I_0)\vee Z\A(I')Z^*$.
Another equivalent requirement to the split property is the existence of faithful normal product states on $\A(I_0)\vee \A(I)'$, see \cite[Proposition 2.2]{MTW18}, see also \cite{DL83} and \cite[below Definition 1.4]{DL84}.
Now, let $I_1, I_2\in \J$ such that $\overline{I_1}\cap \overline{I_2}=\emptyset$ and call $I_3, I_4\in\J$ the two connected components of $I_1'\cap I_2'$. Set $\hat{\A}(E)$ with $E:=I_1\cup I_2$ as the von Neumann algebra $Z(\A(I_3)\vee \A(I_4))'Z^*$. Note that $\A(E):=\A(I_1)\vee \A(I_2)$ is contained in $\hat{\A}(E)$, so that it makes sense to consider the \textit{Kosaki index} $[\hat{\A}(E):\A(E)]$ of this inclusion, see e.g.\ \cite{Kos98}, see \cite[Section 2]{LR95} for a review. Easily adapting the proof of \cite[Proposition 5(a)]{KLM01}, we have that the index $[\hat{\A}(E):\A(E)]$ is independent of the choice of $I_1, I_2\in\J$, provided that $\A$ is strongly additive and split. Accordingly, we call this index the $\mu$-\textit{index} of $\A$ and we denote it by $\mu_\A$.
Then we have the following definition on the line of \cite[Defition 8]{KLM01} in the local case:

\begin{defin}
	An irreducible graded-local conformal net $\A$ is said to be \textit{completely rational} if it is strongly additive, it satisfies the split property and it has finite $\mu$-index.
\end{defin}

\begin{rem}  \label{rem:complete_rationality}
	Let $\A^\Gamma$ be the Bose subnet of an irreducible graded-local conformal net $\A$. Recall that $\A^\Gamma$ can be considered as an irreducible local conformal net acting on $\mathcal{H}^\Gamma$.
	By \cite[Proposition 36]{Lon03} (based in part on an adaptation of \cite[Lemma 22 and Lemma 23]{Lon03} to the graded-local case), we have that $\A$ is completely rational if and only if $\A^\Gamma$ is completely rational. In this case, $\mu_{\A^\Gamma}=[\A:\A^\Gamma]^2\mu_\A = 4\mu_\A$, see the proof of \cite[Proposition 24]{KLM01} and the proof of \cite[Lemma 22]{Lon03}.
\end{rem}	
 
It is worthwhile to recall \cite[Theorem 5.4]{MTW18}, which proves that the split property is automatic in the local case. As a consequence, if the local conformal net has finite $\mu$-index, then also the strong additivity follows, see \cite[Theorem 5.3]{LX04}.  Accordingly, an irreducible local conformal net is completely rational if and only if it has finite $\mu$-index. This fact remains true in the graded-local case as we are going to discuss.

In \cite{Dop82}, it is proved under natural assumptions that if a local net $\B$ on the Minkowski space-time satisfying Haag duality is the fixed-point subnet of a graded-local net $\A$ under the action of a finite abelian group of internal symmetries, then $\B$ is split if and only if $\A$ is split. 
Note that the split property for $\A$ and the split property for $\B$ correspond to assumptions (v) and (v$'$) in page 76 of \cite{Dop82} respectively. The equivalence (v)$\Leftrightarrow$(v$'$) is discussed in page 84 of \cite{Dop82}. The implication  (v)$\Rightarrow$(v$'$) is rather straightforward. The implication  (v$'$)$\Rightarrow$(v) is a consequence of the fact that the action of the finite group can be locally implemented by unitaries in $\B$. This follows from the split property for $\B$ together with \cite[Proposition 3.3]{Dop82}.  A close inspection of the proofs shows that the argument in \cite{Dop82} can be adapted to the setting of M\"{o}bius covariant nets on $S^1$. As a consequence, we have the following. 

\begin{prop}
\label{propsplitgradedmobius}
Let $\A$ be an irreducible graded-local M\"{o}bius covariant net on $S^1$. Then $\A$ satisfies the split property if and only if $\A^\Gamma$ satisfies the split property. 
\end{prop}

Thanks to \cite[Theorem 5.4]{MTW18} we have the following corollary.

\begin{cor}
\label{corsplitgradedconformal}
 Every irreducible graded-local conformal net $\A$ satisfies the split property. 
\end{cor}	

We are now ready to state and prove the following proposition. 

\begin{prop} Let $\A$ be an irreducible graded-local conformal net. Then $\mu_{\A^\Gamma} = 4 \mu_\A$. Moreover, $\A$ is completely rational if and only if it has finite $\mu$-index.
\end{prop}

\begin{proof} 
	By Corollary \ref{corsplitgradedconformal}, $\A$ satisfies the split property. Then by the proof of \cite[Proposition 24]{KLM01}, see also the proof of \cite[Lemma 22]{Lon03}, we have that $[\A:\A^\Gamma]^2\mu_{\A^\Gamma}=[\A:\A^\Gamma]^4\mu_\A$, where $[\A:\A^\Gamma]=2$, see e.g.\ the beginning of \cite[Section 7]{KLM01}. Therefore, if $\mu_\A$ is finite, then $\mu_{\A^\Gamma}=4\mu_\A$ is finite, so that $\A^\Gamma$ is also strongly additive by \cite[Theorem 5.3]{LX04}. Consequently, $\A^\Gamma$ is completely rational and thus $\A$ is too by an adaption of \cite[Lemma 22 and Lemma 23]{Lon03}. 
\end{proof}
 
The notion of complete rationality is strongly related to the representation theory of conformal nets, which we are going to briefly recall, see \cite[Section 2.4 and Section 4.2]{CKL08}, \cite[Section 2]{CHKLX15} and \cite[Section 2]{CHL15} for details, see also \cite{GL96}.
For a \textit{(locally normal) representation} $\pi$ of a graded-local conformal net $\A$, we mean a collection of normal representations $\pi_I:\A(I)\to B(\mathcal{H}_\pi)$ for all $I\in\J$ on a common Hilbert space $\mathcal{H}_\pi$, satisfying $\pi_J\restriction_{\A(I)}=\pi_I$ for all $I,J\in \J$ such that $I\subseteq J$.
Note that any representation of $\A$ restricts to a representation of $\A^\Gamma$.
Let $\J^\R$ be the subset of $\J$ of intervals whose closures do not contain the point $-1\in S^1$. Then a \textit{soliton} $\pi$ of $\A$ is a collection of normal representations $\pi_I:\A(I)\to B(\mathcal{H}_{\pi})$ for all $I\in\J^\R$ on a common Hilbert space $\mathcal{H}_\pi$, satisfying $\pi_J\restriction_{\A(I)}=\pi_I$ for all $I,J\in \J^\R$ such that $I\subseteq J$.
A representation $\pi$ of $\A$ is called \textit{M\"{o}bius} (resp.\ \textit{diffeomorphism}) \textit{covariant} if there exists a strongly continuous (projective) unitary representation $U_\pi$ of $\Mob(S^1)^{(\infty)}$ (resp.\ $\Diff^+(S^1)^{(\infty)}$) on $\mathcal{H}_\pi$ such that for all $\gamma\in\Mob(S^1)^{(\infty)}$ (resp.\ $\Diff^+(S^1)^{(\infty)}$) all $I\in\J$ and all $X\in\A(I)$,
\begin{equation}
	U_\pi(\gamma)\pi_I(X)U_\pi(\gamma)^*=\pi_{\dot{\gamma}I}(U(\gamma)XU(\gamma)^*) \,.
\end{equation}
Moreover, $\pi$ has \textit{positive-energy} if the generator $L_0^\pi$ of the lift to $\Mob(S^1)^{(\infty)}$ of the one-parameter subgroup of rotations $\R\ni t\mapsto U_\pi(r^{(\infty)}(t))$ is positive.
A M\"{o}bius or diffeomorphism covariant representation $\pi$ is \textit{graded} if there exists a self-adjoint unitary operator $\Gamma_\pi$ of $\mathcal{H}_\pi$, commuting with $U_\pi$ and such that for all $I\in\J$ and all $X\in\A(I)$,
\begin{equation}
	\Gamma_\pi\pi_I(X)\Gamma_\pi=\pi_I(\Gamma X\Gamma) \,.
\end{equation}
Note that the local normality of any representation or soliton $\pi$, which we require by definition, is automatic whenever $\mathcal{H}_\pi$ is separable, see e.g.\ \cite[Appendix B]{KLM01}. 
Vice versa, if a representation $\pi$ is cyclic, then $\mathcal{H}_\pi$ is separable. The following proposition is given in \cite[Proposition 12]{CKL08} with the additional assumption that $\pi$ is irreducible.

\begin{prop}  \label{prop:diff_cov_grade-local_net}
	Let $\pi$ be a locally normal representation of an irreducible graded-local conformal net $\A$. 
	Then there exists a unique positive-energy strongly continuous unitary representation $U_\pi$ of $\Mob(S^1)^{(\infty)}$ making
	the restriction of $\pi$ to ${\A^\Gamma}$ M\"{o}bius covariant and such that $U_\pi(\Mob(S^1)^{(\infty)})$ $\subseteq$ $\bigvee_{I\in\J}\pi_I(\Virnet_c(I))$. 
	If $U_\pi(r^{(\infty)}(2\pi))$ is diagonalizable on $\mathcal{H}_\pi$, then there exists a unique positive-energy strongly continuous projective unitary representation $V_\pi$ of $\Diff^+(S^1)^{(\infty)}$ extending $U_\pi$, making $\pi$ a positive-energy diffeomorphism covariant representation of $\A$ and such that (up to a phase) $\pi_I(U(\mathrm{exp}^{(\infty)}(tf)))$ $=$ $V_\pi(\mathrm{exp}^{(\infty)}(tf))$ for all smooth real vector field $f$ on $S^1$ with $\mathrm{supp} f \subseteq I$. Furthermore, $\pi$ is graded by $V_\pi(r^{(\infty)}(2\pi))$ and $V_\pi(\Diff^+(S^1)^{(\infty)})$ $\subseteq$ $\bigvee_{I\in\J}\pi_I(\Virnet_c(I))$. 
\end{prop}

\begin{proof}
	The first part follows directly from \cite[Proposition 3.3 and Theorem 3.8]{Wei06}.
	It follows from the proof of \cite[Proposition 2.2]{Car04} that there exists a strongly-continuous projective unitary representation $V_\pi$ of $\Diff^+(S^1)^{(\infty)}$ extending $U_\pi$ and such that for all $I\in\J$, $V_\pi(\mathrm{exp}^{(\infty)}(tf))=\pi_I(U(\mathrm{exp}^{(\infty)}(tf)))$, up to a phase, for all smooth real vector field $f$ with $\mathrm{supp}f$ contained in $I$. Actually, the representation $V_\pi$ is obtained by integrating a representation $\eta$ by essentially skew-adjoint operators of the Lie algebra of smooth real vector fields on $S^1$ on the common invariant domain 
	$C^\infty(L_0^\pi)$ of smoot vectors for the generator $L^\pi_0$ of the one-parameter group $t \mapsto U_\pi(r^{(\infty)}(t))$, see also \cite{Tol99}.

	Now, we follow an argument in \cite[p.\ 121]{BS90}. 	Let  $I, I_1, I_2 \in\J$ be such that $\overline{I} \subseteq I_1$ and $\overline{I_1} \subseteq I_2$ and let $f$ be any smooth real vector field on $S^1$. Let $f_1$ be a smooth real vector field on $S^1$ with support in $I_2$ and such that $f\restriction_{I_1} = f_1\restriction_{I_1}$. Fix $\delta> 0$ such that $\mathrm{exp}(tf)I \subseteq I_1$ for all $t\in (-\delta,\delta)$. Then, 
$U(\mathrm{exp}^{(\infty)}(t(f-f_1))\in  \A(I_1)'$.  We have, up to a phase, 
$ V_\pi(\mathrm{exp}^{(\infty)}(tf)) = e^{t\eta(f)}$  $= e^{t\eta(f_1) + t\eta(f-f_1)}$. Hence, by the Trotter product formula \cite[Theorem VIII.31]{RSI}, 
again up to a phase, 
$$
V_\pi(\mathrm{exp}^{(\infty)}(tf)) = \lim_{n\to \infty} \left(  e^{\frac{t}{n}\eta(f_1)} e^{\frac{t}{n}\eta(f-f_1)}  \right)^n
$$
in the strong operator topology of $B(\mathcal{H}_\pi)$. As a consequence,
$$ 
  V_\pi(\mathrm{exp}^{(\infty)}(tf))\pi_I(X)V_\pi(\mathrm{exp}^{(\infty)}(tf))^* = 
  \pi_{I_1}(U(\mathrm{exp}^{(\infty)}(tf_1))XU(\mathrm{exp}^{(\infty)}(tf_1))^*) 
$$
for all $X\in \A(I)$ and all $t\in (-\delta,\delta)$. 
On the other hand, again by the Trotter product formula we have, up to a phase,
$$
	U(\mathrm{exp}^{(\infty)}(tf)) = \lim_{n\to \infty} \left(  e^{\frac{t}{n}Y(\nu, f_1)} e^{\frac{t}{n}Y(\nu,f-f_1)}  \right)^n
$$	
in the strong operator topology of $B(\mathcal{H})$. Then 	
$$
   U(\mathrm{exp}^{(\infty)}(tf_1))XU(\mathrm{exp}^{(\infty)}(tf_1))^* =
   U(\mathrm{exp}^{(\infty)}(tf))XU(\mathrm{exp}^{(\infty)}(tf))^*  
$$
for all $X\in \A(I)$ and all $t\in (-\delta,\delta)$,
and thus
$$ 
	V_\pi(\mathrm{exp}^{(\infty)}(tf))\pi_I(X)V_\pi(\mathrm{exp}^{(\infty)}(tf))^* = 
	\pi_{I_1}(U(\mathrm{exp}^{(\infty)}(tf))XU(\mathrm{exp}^{(\infty)}(tf))^*) 
$$       
for all $X\in \A(I)$ and all $t\in (-\delta,\delta)$. It follows that for any $X\in\A(I)$, the set of the real numbers $t$ for which 
$$
  V_\pi(\mathrm{exp}^{(\infty)}(tf))\pi_I(X)V_\pi(\mathrm{exp}^{(\infty)}(tf))^* = \pi_{\mathrm{exp}(tf)I}(U(\mathrm{exp}^{(\infty)}       
  (tf))XU(\mathrm{exp}^{(\infty)}(tf))^*)
$$
is open and closed and hence it coincides with $\mathbb{R}$. Accordingly, if $H$ denotes the subgroup of $\Diff^+(S^1)^{(\infty)}$ generated by the exponentials we find that   
$$
 V_\pi(\gamma)\pi_I(X)V_\pi(\gamma)^* = \pi_{\dot{\gamma}I}(U(\gamma)XU(\gamma)^*) 
$$
for all $\gamma \in H$ and all $X\in\A(I)$. Now let $\tilde{p}:  \Diff^+(S^1)^{(\infty)}  \to \Diff^+(S^1)$ be the covering map. Its kernel is generated by $r^{(\infty)}(2\pi)$ and hence it is a subroup of $H$. Since $\Diff^+(S^1)$ is a simple group \cite[Remark 1.7]{Milnor84}, it is generated by exponentials. It follows that $\tilde{p}(H) = \Diff^+(S^1)$ that, together with $\mathrm{Ker} (\tilde{p}) \subseteq H$, implies that $H= \Diff^+(S^1)^{(\infty)}$. Then the diffeomorphism covariance of the representation $\pi$ of $\A$ follows.
\end{proof}

\begin{rem}  \label{rem:irrationality_local_net_from_graded_net}
	By \cite[Theorem 4.9]{LX04}, the complete rationality of a local conformal net is equivalent to say that it has a finite number of equivalence classes of irreducible representations and all of them have finite index.
	By Proposition \ref{prop:diff_cov_grade-local_net}, for every representation $\pi$ we can consider the infimum $h_\pi\geq0$ of the spectrum of $L_0^\pi$. 
	Then if there exists a family $\mathscr{E}$ of representations of $\A$ such that the set $\{h_\pi\mid \pi\in\mathscr{E}\}$ is infinite, then $\A^\Gamma$ cannot be completely rational. Therefore, also $\A$ is not completely rational.
\end{rem}

Now, we are ready to discuss the complete rationality of minimal $W$-algebra net and their relation with the representation theories of the $W$-algebras they come from. We first list some known facts and then we move to our main result on the representation theory of minimal $W$-algebra nets. 
In what follows and in the subsequent sections, we will use graded tensor products of VOSA modules and of conformal net representations, see e.g.\ \cite[Sections 4.6--4.7]{FHL93} and \cite[Section 2.6]{CKL08} respectively.

In \cite[Section 11]{KMP23} and \cite[Section 6]{AKMP24}, it is shown that a unitary minimal $W$-algebra $W^k(\g,f)$ is not rational if $k$ is not collapsing and it is not a Virasoro VOA, a $N=1$ or $N=2$ super-Virasoro VOSA with central charges in their corresponding unitary discrete series, see \eqref{eq:unitary_series_vir}, \eqref{eq:unitary_series_N=1} and \eqref{eq:unitary_series_N=2} respectively. 
It is known that unitary affine VOAs are strongly rational if and only if their corresponding local conformal nets are completely rational. 
Indeed, if $V$ is a unitary affine VOA, then its weight-$1$ subspace $V_1$ is a reductive Lie algebra, see \cite[Theorem 4.3]{AL17}. Then $V$ is a tensor product of some copies of the ($c=1$) Heisenberg VOA and of some unitary affine VOAs coming from simple finite-dimensional complex Lie algebras, see \cite[Remark 5.7c]{Kac01}.
The latter VOAs and their tensor products are strongly rational by \cite[Theorem 3.7 and Proposition 3.3]{DLM97} and the corresponding conformal nets are completely rational by \cite[Theorem 5.1]{Ten24} and \cite[Theorem I]{Gui}. 
The Heisenberg VOA is not rational by e.g.\ \cite[Section 6.3]{LL04} and its corresponding conformal net is not completely rational by e.g.\ \cite[Section III]{BMT88}. 
Then the claim follows.
Moreover, it is already known that Virasoro conformal nets are completely rational if and only if their central charges are in the unitary discrete series \eqref{eq:unitary_series_vir}, see Section \ref{subsec:virasoro_nets} for details and references.

The $N=3,4$ and the big $N=4$ super-Virasoro nets can be obtained from the $W$-algebras $W_k(\g,f)$ for $\g=\mathfrak{spo}(2|3)$, $\g= \mathfrak{psl}(2|2)$ and $\g=D(2,1;a)$ with $a\in\C\backslash\{-1,0\}$ respectively. Details of that are given in Sections \ref{subsec:N=3_super-virasoro_nets}--\ref{subsec:big_N=4_super-virasoro_nets} respectively.
To prove the following theorem, beyond the above discussion, we first construct for each of the $N=1,2,3,4$ and the big $N=4$ super-Virasoro nets, an infinite family of non-equivalent representations for suitable values of the central charge, see Sections \ref{subsec:N=1_super-virasoro_nets}--\ref{subsec:big_N=4_super-virasoro_nets} respectively.
Then we apply the argument in Remark \ref{rem:irrationality_local_net_from_graded_net}.

\begin{theo}
	Let $W_k(\g,f)$ be a unitary minimal $W$-algebra. If $W_k(\g,f)$ is (strongly) rational, then $\A_{W_k(\g,f)}$ is completely rational.
	The converse is also true whenever $\g\in\{\mathfrak{sl}_2,\, \mathfrak{spo}(2|1),\, \mathfrak{spo}(2|2),\, \mathfrak{spo}(2|3),\, \mathfrak{psl}(2|2)\}$ or $k$ is a collapsing level.
	Furthermore, the $N=0,1,2,3,4$ and the big $N=4$ super-Virasoro nets are completely rational if and only if their corresponding VOSAs are strongly rational.
\end{theo}

We conjecture that the second statement of the above theorem remains true for all unitary minimal $W$-algebras, that is:

\begin{con}
	Let $W_k(\g,f)$ be a unitary minimal $W$-algebra. Then $W_k(\g,f)$ is strongly rational if and only if $\A_{W_k(\g,f)}$ is completely rational.
\end{con}

\begin{rem} It is worthwhile to remark that the method used to construct the representations for the $N=1$ and $N=2$ super-Virasoro nets in Section \ref{subsec:N=1_super-virasoro_nets} and Section \ref{subsec:N=2_super-virasoro_nets} respectively, accidentally give alternative proofs of the unitarity of the ``vacuum'' representations, that is the ones with zero lowest weight, of the Neveu-Schwarz algebra with central charge $c\geq \frac{3}{2}$ and of the $N=2$ super-Virasoro algebra with central charge $c\geq 3$, see Corollary \ref{cor:unitarity_vacuum_representations_NS} and Corollary \ref{cor:unitarity_vacuum_representations_N2} respectively.
Actually, the same happens for the Virasoro algebra with central charge $c\geq 1$ in \cite[Proposition 8]{CTW23}.
\end{rem}

\subsection{Virasoro nets}  \label{subsec:virasoro_nets}

The simple minimal $W$-algebras $W_k(\g,f)$ with $\g=\mathfrak{sl}_2$ are the (simple) \textit{Virasoro VOAs} $L(c(k))$ with central charges $c(k)$ as in the left hand side of \eqref{eq:unitary_series_vir}, see \cite[Section 5]{KRW03} and \cite[Section 8.1]{KW04}. 
These are usually obtained through representations of the Virasoro algebra $\Vir$ defined by \eqref{eq:virasoro_cr}, see \cite[Example 4.10]{Kac01} and \cite[Section 4.1]{DL14}.
The \textit{unitary series} for $L(c(k))$ is as follows:
\begin{equation}  \label{eq:unitary_series_vir}
	c(k)=1-\frac{6(k+1)^2}{k+2}
	\left\{\begin{array}{ll}
		 =1-\frac{6}{p(p+1)} \,, &  k=\frac{1}{p}-1 \,,\,\, p\in \Z_{\geq 2} \\
		 =1 \,,&  k=-1 \\
		 > 1 \,,& k<-2
	\end{array}	\right. 
\end{equation}
One refers to the set of values of $c(k)$ as in the first row in \eqref{eq:unitary_series_vir} as the \textit{unitary discrete series}, for which it is known that $L(c(k))$ is strongly rational, see \cite[Theorem 3.13]{DLM97}. On the contrary, $L(c(k))$ is not rational for $c(k)\geq 1$. Classical references are \cite{FQS85, GKO85, GKO86, TK86}, \cite[Section 4]{KW86} and \cite[Lecture 3.3]{KR87}.

The corresponding \textit{Virasoro net} $\Virnet_{c(k)}:=\A_{L(c(k))}$, see \cite[Example 8.4]{CKLW18} and references therein, is known to be completely rational for $c(k)$ in the unitary discrete series, see \cite[Corollary 3.4]{KL04}, whereas it is not completely rational for $c(k)\geq 1$, see \cite[Section 4, p.\ 124]{BS90}, see also \cite{Car04, Wei17}.

It is worthwhile to mention that, in both settings, the extensions of this class of models had been completely classified when $c(k)$ is in the unitary discrete series, see \cite[Section 4]{KL04}, \cite[Section 4]{DL15}, \cite[Corollary 2.22]{Gui22}, \cite[Theorem 6.1]{CGGH23} and \cite[Corollary 0.6]{Guib}; whereas just partial results have been reached for the case with $c(k)=1$, see \cite{DM04, DJ13, DJ14, DJ15, Lin17} and \cite[Section IV]{BMT88}, \cite[Section 3]{Car04}, \cite[Section 4]{Xu05}, see also \cite{CGH19}.

\subsection{\texorpdfstring{$N=1$ super-Virasoro nets}{N=1 super-Virasoro nets}}  \label{subsec:N=1_super-virasoro_nets}
The simple minimal $W$-algebras $W_k(\g,f)$ with $\g=\mathfrak{spo}(2|1)$ are the (simple) $N=1$ \textit{super-Virasoro VOSAs} $V^{c(k)}(NS)$ with central charges $c(k)$ as in the left hand side of \eqref{eq:unitary_series_N=1}, see \cite[Section 6]{KRW03} and \cite[Section 8.2]{KW04}. These models are usually obtained through representations of the Neveu-Schwarz algebra $NS$ defined by \eqref{eq:NS_cr}, see \cite[Lemma 5.9]{Kac01} and \cite[Section 2.2]{AL17}.
The \textit{unitary series} for $V^{c(k)}(NS)$ is as follows:
\begin{equation}  \label{eq:unitary_series_N=1}
	c(k)=\frac{3}{2}-\frac{12(k+1)^2}{2k+3}
	\left\{\begin{array}{ll}
		=\frac{3}{2}\left(1-\frac{8}{p(p+2)}\right) \,, &  k=\frac{1}{p}-1 \,,\,\, p\in \Z_{\geq 2} \\
		=\frac{3}{2} \,,&  k=-1 \\
		>\frac{3}{2} \,,& k<-\frac{3}{2}
	\end{array}	\right. 
\end{equation}
As in the Virasoro case, $V^{c(k)}(NS)$ turns out to be strongly rational for $c(k)$ in its \textit{unitary discrete series}, that is the first row in \eqref{eq:unitary_series_N=1}, whereas it is not rational for $c(k)\geq \frac{3}{2}$, see \cite{FQS85, GKO86}, \cite[Section 5]{KW86}.

We denote the corresponding $N=1$ \textit{super-Virasoro net} by $\SVirnet_{c(k)}:=\A_{V^{c(k)}(NS)}$, see \cite[Example 7.5]{CGH} and references therein. Then the complete rationality of $\SVirnet_{c(k)}$ with $c(k)$ in the unitary discrete series was established in \cite[Section 6.4]{CKL08}. Moreover, a complete classification of irreducible graded-local conformal net extensions of $\SVirnet_{c(k)}$ with $c(k)$ in the unitary discrete series was obtained in \cite[Section 7]{CKL08}, which has been recently transposed in the VOSA setting by \cite[Theorem 6.2]{CGGH23}.
To our knowledge, the failure of complete rationality of the case $c(k)\geq \frac{3}{2}$ has not been proved, although commonly accepted. In the following, we will prove this fact, constructing an infinite number of irreducible representations for $\SVirnet_{c(k)}$ with $c(k)\geq \frac{3}{2}$, integrating unitary representations of $NS$. The idea, divided into two steps in Subsection \ref{subsubsec:step1} and Subsection \ref{subsubsec:step2}, follows mostly the contents of \cite{CTW23} and \cite[Section 4]{BS90} respectively, cf.\ \cite[Section 4.4]{CTW22}. 

\subsubsection{Formal series and fields}
\label{subsubsec:formal_series_fields}

Preliminarily, we introduce the forthcoming notations following \cite[Section 2.1, Section 2.2, Section 2.4]{CTW23}. 
Let $W=W_\parzero\oplus W_\parone$ be a vector superspace.
A \textit{formal series} in $W$ is an expression  $A(z)=\sum_{n\in\half\Z}A_nz^{-n}$ in the \textit{formal variable} $z$ with $A_n\in\End(W)$ for all $n\in\half\Z$. We say that $A(z)$ has \textit{parity} $q\in \Z/2\Z$ if for all $n\in \half\Z$, $A_nV_p\subseteq V_{p+q}$ for all $p\in\Z/2\Z$.
We refer to the $A_n$'s as the \textit{Fourier coefficients} of $A(z)$.
Moreover, the two expressions
\begin{equation}
	\partial_zA(z):=\sum_{n\in\half\Z}(-n)A_nz^{-n-1}
	\quad\mbox{and}\quad
	A'(z):=iz\partial_z A(z) 
\end{equation}
are called the \textit{formal derivative} and the \textit{circle derivative} of $A(z)$ respectively.
We say that $A(z)$ is a \textit{field} if for all $w\in W$, there exists $N\in\Z$ such that $A_nw=0$ for all $n> N$.
Note that the formal and the circle derivatives of a field are fields too.  
The \textit{normally ordered product} $:A(z)B(z):$ of two fields $A(z)$ and $B(z)$ with parity $p$ and $q$ respectively is defined by
\begin{equation}
	:A(z)B(z): = A(z)_+ B(z)+(-1)^{pq}B(z)A(z)_- 
\end{equation}
where $A(z)_+:=\sum_{n\leq 0}A_nz^{-n}$ and $A(z)_-:=\sum_{n>0}A_nz^{-n}$.
Suppose that $W$ is equipped with a scalar product.
Then we say that a formal series $A(z)$ is \textit{symmetric} if
\begin{equation}
	A(z)^\dagger:= \sum_{n\in\half\Z}A_n^\dagger z^n= A(z)
\end{equation}
where for any operator $A:W\to W$, $A^\dagger$ denotes the \textit{adjoint} of $A$, that is an operator $B:W\to W$ (which may not exists) such that
\begin{equation}
	\forall \Psi_1,\Psi_2\in W\qquad
	(A\Psi_1|\Psi_2)=(\Psi_1|B\Psi_2) \,.
\end{equation}
In other words, $A(z)$ is symmetric if and only if $A_n^\dagger=A_{-n}$ for all $n\in\half\Z$. Note that if $A(z)$ is symmetric, then also its formal and circle derivatives are.
Consider \textit{formal polynomials}, allowing half-integer powers of the formal variable $z$, with their complex conjugates: 
$$
r(z):=\sum_{\substack{n\in\half\Z \\ \abs{n}<N}}r_nz^{-n}
\quad\mbox{ and }\quad
\overline{r}(z):=\sum_{\substack{n\in\half\Z \\ \abs{n}<N}}\overline{r_{-n}}z^{-n}
$$ 
where $N\in \Zpluseq$. Then as formal series, we have that
\begin{equation}
	\left(r(z)A(z)\right)^\dagger=\overline{r}(z)A(z)^\dagger \,.
\end{equation}
In particular, if $\overline{r_{-n}}=r_n$ for all $n$, then $r(z)A(z)$ is symmetric if $A(z)$ is.
Finally, for fields in a vector superspace equipped with a scalar product, we can define \textit{energy bounds} in the obvious way following Definition \ref{defin:energy_bounds}.

\subsubsection{Representations of the Neveu-Schwarz algebra}
\label{subsubsec:step1}

Consider the graded tensor product $V:=M(1)\hat{\otimes} F$, where $M(1)$ is the Heisenberg VOA and $F$ is the real free fermion VOSA, see references in Remark \ref{rem:unitary_range_and_collapsing} and Example \ref{ex:real_free_fermion} respectively. The unitary structure on $V$ is the one induced by the ones of $M(1)$ and $F$ as explained in Section \ref{section:preliminaries}.
By \cite[Example 5.9a]{Kac01}, there is a family in $V$ of Virasoro vectors with associated superconformal vectors, depending on a parameter $\varsigma\in \C$ and giving a family of representations of $NS$ on the vector superspace $V$ with central charge $c(\varsigma)=\frac{3}{2}-3\varsigma^2$:
\begin{equation}  \label{eq:bosonxfermion_representations_NS}
	\begin{split}
		\nu^\varsigma 
		&:=
		\half\left((\bos_{-1})^2+\fer_{-\frac{3}{2}}\fer_{-\frac{1}{2}}+\varsigma \bos_{-2}\right)\Omega \\
		\tau^\varsigma
		&:=
		\left(\bos_{-1}\fer_{-\half}+\varsigma \fer_{-\frac{3}{2}}\right)\Omega 
	\end{split}
\end{equation}
where $\bos\in M(1)$ and $\fer\in F$ are Hermitian primary vectors of conformal weight $1$ and $\half$ with norm one respectively. Note also that $\nu^\varsigma$ and $\tau^\varsigma$ are not Hermitian vectors unless $\varsigma=0$, which is the only case where \eqref{eq:bosonxfermion_representations_NS} gives a unitary representation of $NS$. Nevertheless, we get a unitary embedding of $V^\frac{3}{2}(NS)$ into $V$, making $V$ a unitary $N=1$ superconformal VOSA.
Accordingly, set the standard notation $\nu:=\nu^0$ and $\tau:=\tau^0$, $Y(\nu, z)=\sum L_mz^{-m-2}$ and $Y(\tau,z)=\sum G_n z^{-n-\frac{3}{2}}$.
Recall also that $V$ is strongly graded-local and thus $\mathrm{SVir}_\frac{3}{2}$ is contained in the irreducible graded-local conformal net $\A_V=\A_{\Uone}\hat{\otimes} \F$ as covariant subnet, making $\A_V$ a $N=1$ superconformal net (this is the one constructed in the proof of \cite[Lemma 3.3]{Hil15}). 
We are interested in unitary representations of the Lie superalgebra $NS$, which can exist only for $c(\varsigma)>0$ as we have recalled just above. Then we are forced to consider only purely complex values for $\varsigma$ and thus to modify $\nu^\varsigma$ and $\tau^\varsigma$ to obtain the desired unitary representations of $NS$.  

We rewrite \eqref{eq:bosonxfermion_representations_NS} in terms of fields, see Subsection \ref{subsubsec:formal_series_fields}.
Choose $\varsigma=2i\kappa$ with $\kappa\in \R$ and let $J(z):=zY(\bos,z)$, $\Phi(z):=z^\half Y(\fer,z)$, $J_m:=\bos_m$ for all $m\in\Z$ and $\Phi_n:=\fer_n$ for all $n\in\Z-\half$.
Note that the fact that $\bos$ and $\fer$ are Hermitian is equivalent to be symmetric for $J(z)$ and $\Phi(z)$.
Using the normally ordered product, define the following \textit{shifted normal powers}:
\begin{equation} \label{eq:shifted_normal_powers}
	\begin{split}
		:J^2:(z)&:=z^2:Y(\bos,z)Y(\bos,z):  \\
		:\partial\Phi\Phi:(z)&:=z^2:\partial_z( Y(\fer,z)) Y(\fer,z): \\
		:J\Phi:(z)&:= z^\frac{3}{2}:Y(\bos,z)Y(\fer,z): \,.
	\end{split}
\end{equation}
Note that these shifted normal powers are symmetric as $J(z)$ and $\Phi(z)$ are.
Then for all $\kappa\in\R$, we set the following fields from \eqref{eq:bosonxfermion_representations_NS}:
\begin{equation}  \label{eq:bosonxfermion_representations_NS_fields}
	\begin{split}
		\widetilde{T}(z;\kappa)
			&:=
		z^2Y(\nu^\varsigma,z) =
		\half:J^2:(z)+\half:\partial\Phi\Phi:(z)-i\kappa \left(J(z)+iJ'(z)\right) \\
		\widetilde{G}(z;\kappa)
			&:=
		z^\frac{3}{2}Y(\tau^\varsigma,z)=
		:J\Phi:(z)-i\kappa\left(\Phi(z)+2i\Phi'(z)\right) 
	\end{split}
\end{equation}
whose Fourier coefficients, namely 
$$
\left\{\left. \widetilde{L}_{\kappa,m}\,,\,\, \widetilde{G}_{\kappa,n} \,\right|\, m\in\Z\,,\,\, n\in\Z-\half \right\}
$$ 
give rise to a lowest weight representation of $NS$ with central charge $c(\kappa):=\frac{3}{2}+12\kappa^2$ and lowest weight vector $\Omega$ of lowest weight $h=0$, where $\widetilde{L}_{\kappa,0}\Omega=h\Omega$.
Accordingly, we set $T(z):=\widetilde{T}(z;0)=z^2Y(\nu,z)$ and $G(z):=\widetilde{G}(z;0)=z^\frac{3}{2}Y(\tau,z)$.

Let $f(z)=\sum_{n\in\Z}c_nz^{-n}$ be a scalar-valued field, that is  $c_n\in\C$ for all $n\in\Z$ and $c_n=0$ for all $n>N$ for some $N\in\Z$.
Following \cite[Section 3.2]{CTW23}, we have that the transformation $\varphi_f$, sending $J(z)\mapsto J(z)+f(z)$ and $\Phi(z)\mapsto \Phi(z)$, preserves the commutation relations involving the Fourier coefficients of $J(z)$ and $\Phi(z)$.
Then $\varphi_f$ can be understood as a composition of the representation on $V$ and of an automorphism of the infinite-dimensional Lie superalgebra given by the generators $\{J_n \,,\,\, \Phi_m\mid n\in\Zplus\,,\,\, m\in\Z-\half\}$.
In particular, if $N=0$ and $\Psi$ is a lowest weight vector for $J(z)$ with lowest weight $q\in \C$ (that is $J_0\Psi=q\Psi$ and $J_n\Psi=0$ for all $n>0$), then $\Psi$ is a lowest weight vector for $J(z)+f(z)$ too with lowest weight $q+c_0$.
It also follows that if we modify $\widetilde{T}(z;\kappa)$ and $\widetilde{G}(z;\kappa)$ inserting such transformed $\varphi_f(J(z))$ in place of $J(z)$ in their defining formulae \eqref{eq:bosonxfermion_representations_NS_fields}, then we still have a representation of $NS$ with the same central charge and lowest vector, but with possibly different lowest weight. 
In particular, inspired by \cite[Section 4, p.\ 123--124]{BS90}, we will be mostly interested in transformations of the form $\varphi_{\kappa,\eta}:=\varphi_{f_{\kappa,\eta}}$ with
\begin{equation}
	f_{\kappa,\eta}(z):=\kappa\widetilde{\rho}(z)+\eta
	\,,\quad
	\widetilde{\rho}(z):= -i\frac{z-1}{z+1}
\end{equation}
for all $\kappa,\eta\in\C$, where the meromorphic function $\widetilde{\rho}(z)$ is interpreted as the formal series arising from its expansion around the origin, that is
\begin{equation}  \label{eq:series_expansion_rho}
	\widetilde{\rho}(z)=i\left(1+2\sum_{n= 1}^{+\infty}(-1)^nz^n \right)
	\,.
\end{equation}
Note that $\widetilde{\rho}(z)$ satisfies the differential equation
\begin{equation}  \label{eq:rho_differential_eq}
	\frac{\widetilde{\rho}(z)^2}{2}+\half-iz\partial_z\widetilde{\rho}(z)=0 \,.
\end{equation} 
 
With the same $\kappa\in\R$ as in \eqref{eq:bosonxfermion_representations_NS_fields}, we first perform the transformation $\varphi_{-\kappa,i\kappa}$, to get:
\begin{equation}
\label{eq:bosonxfermion_representations_NS_fields_BS}
\begin{split}
	T_\kappa(z)
	&:=
	\half:J^2:(z)+\half:\partial\Phi\Phi:(z)+\kappa \left(J'(z)-\widetilde{\rho}(z)J(z)\right) \\
	G_\kappa(z)
	&:=
	:J\Phi:(z)+\kappa\left(2\Phi'(z)-\widetilde{\rho}(z)\Phi(z) \right)
\end{split}
\end{equation}
where we have used \eqref{eq:rho_differential_eq} in the calculation.
According to what said above, for all $\kappa\in\R$ the Fourier coefficients 
\begin{equation}  \label{eq:bosonxfermion_representations_NS_fields_BS_fourier}
	\begin{split}
		L_{\kappa, m}
			&= 
		L_m-i\kappa\left[(1+m)J_m+2\sum_{j<0}(-1)^jJ_{m-j}\right]
		,\quad m\in\Z
		\\
		G_{\kappa,n}
			&=  
		G_n -i\kappa\left[(1+2n)\Phi_n+2\sum_{j<0}(-1)^j\Phi_{n-j}\right] ,\quad n\in\Z-\half
	\end{split}
\end{equation}
of the fields in \eqref{eq:bosonxfermion_representations_NS_fields_BS} respectively give a representation of $NS$ with central charge $c(\kappa)$ and lowest weight vector $\Omega$ of lowest weight $h=0$. 
We highlight that this is the version for $NS$ of \cite[Eq.\ (4.6)]{BS90}, which will be useful in Subsection \ref{subsubsec:step2}.
However, note that this representation is not unitary as the fields $T_\kappa(z)$ and $G_\kappa(z)$ are not symmetric due to the series expansion \eqref{eq:series_expansion_rho}. Nevertheless, we note that $\widetilde{\rho}(z)\in\R$ for all $z\in\pSone$ and thus a weak version of symmetry holds, that is for every trigonometric polynomial $p(z)$ such that $p(-1)=0$, we have that for all $\kappa\in\R$
\begin{equation}  \label{eq:weak_symmetry_NS}
	\left(p(z)T_\kappa(z)\right)^\dagger=\overline{p}(z)T_\kappa(z)
	\quad\mbox{ and }\quad
	\left(p(z)z^\half G_\kappa(z)\right)^\dagger=\overline{p}(z)z^{-\half}G_\kappa(z) 
\end{equation}
as formal series.
Now, we prove that this weak symmetry is actually enough for the Fourier coefficients $L_{\kappa, m}$ and $G_{\kappa,n}$ to give rise to a unitary representation of $NS$ from the vacuum vector $\Omega$. This is a straightforward consequence of the following general result:

\begin{theo}  \label{theo:unitarity_BS_representation_NS}
	Let $\left\{L_m \,,\,\, G_n\mid m\in\Z \,,\,\,n\in\Z-\half\right\}$ be a representation of the Neveu-Schwarz algebra $NS$ on a vector superspace $V$ with central charge $c\geq \frac{3}{2}$ such that: 
	\begin{itemize}
		\item [(i)] there exists a cyclic lowest weight vector $\Omega$ such that $L_0\Omega=0$;
		
		\item [(ii)] there exists a scalar product $\scalar$ such that the fields $T(z)=\sum_{m\in\Z} L_mz^{-m}$ and $G(z)=\sum_{n\in\Z-\half} G_nz^{-n}$ satisfy
		\begin{equation}  \label{eq:weak_symmetry_NS_theo}
			\left(p(z)T(z)\right)^\dagger=\overline{p}(z)T_\kappa(z)
			\quad\mbox{ and }\quad
			\left(p(z)z^\half G(z)\right)^\dagger=\overline{p}(z)z^{-\half}G(z) 
		\end{equation}
		for all trigonometric polynomials $p(z)$ such that $p(-1)=0$. 
	\end{itemize}
	Then $L_m^\dagger=L_{-m}$ for all $m\in\Z$ and $G_n^\dagger=G_{-n}$ all $n\in\Z-\half$, that is the representation is unitary.
	
	Furthermore, this representation has a structure of simple unitary VOSA with central charge $c$ (unitarily) isomorphic to the $N=1$ super-Virasoro VOSA $V^c(NS)$.
\end{theo}

\begin{proof}
	For the first statement, we adapt the proof of \cite[Theorem 6]{CTW23} to our setting.
	First, we give some preliminary result. 
	Suppose that $p(z)=\sum_{\abs{n}<N}p_nz^{-n}$ for some positive integer $N$.
	Then looking at the Fourier coefficients of the zero power of \eqref{eq:weak_symmetry_NS_theo}, we have that
	\begin{equation}  \label{eq:zero_mode_dagger}
		\begin{split}
		\left(\sum_{\abs{n}<N}p_{-n}L_n\right)^\dagger
			&=
		\sum_{\abs{n}<N}\overline{p_{-n}}L_{-n}
		\\
		\left(\sum_{\abs{n}<N}p_{-n}G_{n+\half}\right)^\dagger
			&=
		\sum_{\abs{n}<N}\overline{p_{-n}}G_{-n-\half} \,.
		\end{split}
	\end{equation}
	Therefore, choosing $p(z)=z^{n}-(-1)^{n-m}z^m$ for all $n,m\in\Z$, we have that
	\begin{align}  
		\left(L_n-(-1)^{n-m}L_m\right)^\dagger
		&=L_{-n}-(-1)^{n-m}L_{-m}  
		\label{eq:dagger_virasoro_field_coefficients} \\
		\left(G_{n+\half}-(-1)^{n-m}G_{m+\half} \right)^\dagger
		&=G_{-n-\half}-(-1)^{n-m}G_{-m-\half}
		\label{eq:dagger_superconformal_field_coefficients}
	\end{align}
	By the proof of \cite[Lemma 5]{CTW23}, we have that $L_{-1}\Omega=0$. Then by the $NS$ commutation relations \eqref{eq:NS_cr}, we get that
	\begin{equation}
		G_{-\half}\Omega=[G_\half,L_{-1}]\Omega=0 \,.
	\end{equation}
	
	Now, as $\Omega$ is a cyclic lowest weight vector, we have that vectors of the following type generate the whole representation:
	\begin{equation}  \label{eq:generic_generator_NS_ordered}
		L_{-m_1}\cdots L_{-m_s}G_{-n_1}\cdots G_{-n_t}\Omega
	\end{equation}
	for $s,t\in\Zpluseq$, integers $0<m_s\leq\cdots\leq m_1$ and half-integers $0<n_t\leq\cdots\leq n_1$.
	Our first aim is to prove that $L_0^\dagger=L_0$. We proceed proving that $\scalar$ makes the different eigenspaces of $L_0$ orthogonal to each other. To this goal, it is sufficient to prove that given two vectors
	\begin{equation}  \label{eq:generic_generators_NS}
		\Psi:=L_{-m_1}\cdots L_{-m_s}G_{-n_1}\cdots G_{-n_t}\Omega
		\quad\mbox{ and }\quad
		\Psi':=L_{-m_1'}\cdots L_{-m_{s'}'}G_{-n_1'}\cdots G_{-n_{t'}'}\Omega
	\end{equation}
	for $s,s',t,t'\in\Zpluseq$, any (non-necessarily ordered) positive integers $m_1,\cdots,m_t$, $m_1', \cdots, m_{t'}'$ and any (non-necessarily ordered) positive semi-integers $n_1,\cdots, n_s$, $n_1',\cdots, n_{s'}'$, we have that $\langle \Psi,\Psi'\rangle=0$ whenever $d\not=d'$, where $d:=\sum_j m_j+\sum_i n_i$ and $d':=\sum_j m_j'+\sum_i n_i'$ are the eigenvalues of $L_0$ of $\Psi$ and $\Psi'$ respectively. 
	For any vector of the type of $\Psi$, we call the quantity $4s+3t$ the \textit{grade} of $\Psi$.
	Then we adopt the following strategy, based on an induction argument on the \textit{total grade value} 
	\begin{equation}
		\mathrm{gr}:=4(s+s')+3(t+t') \,.
	\end{equation}
	If $\mathrm{gr}=0$, the only possibility is that $\Psi=\Psi'=\Omega$ and our statement is trivially true. Suppose by the inductive hypothesis that it is true for $\mathrm{gr}\leq j$ for a fixed $j\geq 0$, then we have to prove that the statement holds also for $\mathrm{gr}=j+1$. We can further assume that $d\not=d'$ otherwise we trivially have nothing to prove.
	
	We proceed splitting the problem into two cases: $s+s'>0$ and $s+s'=0$. The first case works very similarly to the \textit{Case 1} of the proof of \cite[Theorem 6]{CTW23} as the only ingredients we need are \eqref{eq:dagger_virasoro_field_coefficients}, the fact that $L_{-1}\Omega=0$ and that the $NS$ commutation relations \eqref{eq:NS_cr} ``lower the grade''; whereas we need to adapt the argument for the second case. Accordingly, suppose that $s+s'=0$, so that $\Psi=G_{-n_1}\cdots G_{-n_t}\Omega$ and $\Psi'=G_{-n_1'}\cdots G_{-n_{t'}'}\Omega$. Then we can further suppose that $t+t'>0$ otherwise $\Psi=\Psi'=\Omega$, which is not the case. Without loss of generality, assume that $t>0$. Note that $\Psi\not=(G_{-\half})^r\Omega$ for all positive integers $r$ otherwise it is $0$. Moreover, if $n_1=\half$, using the $NS$ commutation relations \eqref{eq:NS_cr} to ``move'' $G_{-\half}$ to the right, we can rewrite $\Psi$ as a linear combination of the vector $G_{-n_2}\cdots G_{-n_t}G_{-\half}\Omega=0$ and vectors of kind $L\cdots LG\cdots G\Omega$, where we are omitting the indexes, which are negative, all having the same weight $d$ but a strictly smaller grade. Then we are in the first case and with $\mathrm{gr}<j+1$, so that the conclusion for $\mathrm{gr}=j+1$ follows from the inductive hypothesis.
	Therefore, suppose that $n_1>\half$	and define $\xi$ such that $G_{-n_1}\xi=\Psi$. Set $A$ to be equal to
	\begin{equation}  \label{eq:operator_A_for_superconformal}
		\left\{\begin{array}{lr}
			G_{n_1}-(-1)^{n_1-\half}G_\half
			=
			\left(G_{-n_1}-(-1)^{n_1-\half}G_{-\half}\right)^\dagger
			&\mbox{if}\,\,\, d-n_1+\half\not=d'\\
			G_{n_1}+(-1)^{n_1-\half}G_{-\half}
			=
			\left(G_{-n_1}+(-1)^{n_1-\half}G_\half\right)^\dagger 
			&\mbox{if}\,\,\, d-n_1+\half=d'
		\end{array}\right.
	\end{equation}
	where the two equalities hold by \eqref{eq:dagger_superconformal_field_coefficients}. 
	Then we have that
	\begin{equation}
		( \Psi|\Psi' )
		= ( (A^\dagger \pm (-1)^{n_1-\half}G_{\mp\half})\xi |\Psi' )
		= ( \xi|A\Psi' )
		+( \pm (-1)^{n_1-\half}G_{\mp\half} \xi|\Psi' )
	\end{equation}
	where the upper signs are chosen when $d-n_1+\half\not=d'$ and the downer otherwise.
	Note that $A$ and $G_{\pm\half}$ annihilate $\Omega$ and using the $NS$ commutation relations \eqref{eq:NS_cr} as done before, we can rewrite both pairings above as linear combinations of pairing whose terms have $\mathrm{g}<j+1$ and different eigenvalues thanks to the choice on the sign made on $A$. This means that we are in charge to apply the inductive hypothesis to get the desired result.
	
	To move to the conclusion of this first part, we have just proved that $L_0^\dagger=L_0$, which implies that $L_m^\dagger=L_{-m}$ for all $m\in\Z$ thanks to \eqref{eq:dagger_virasoro_field_coefficients}, choosing $m=0$ there.
	By \eqref{eq:dagger_superconformal_field_coefficients}, $(G_\half+G_{-\half})^\dagger=G_\half+G_{-\half}$ and thus using the $NS$ commutation relations \eqref{eq:NS_cr}, we get that
	\begin{equation}
		G_\half^\dagger
		=[L_1,G_\half+G_{-\half}]^\dagger
		=[(G_\half+G_{-\half})^\dagger, L_1^\dagger]
		=[G_\half+G_{-\half}, L_{-1}]
		=G_{-\half} \,.
	\end{equation}
	Therefore, using \eqref{eq:dagger_superconformal_field_coefficients} with $m=0$, we also have that $G_n^\dagger=G_{-n}$ for all $n\in\Z-\half$.
	
	The vertex superalgebra structure is assured by \cite[Theorem 4.5]{Kac01}, where the vector superspace $V$ is obviously the one generated by vectors of kind \eqref{eq:generic_generator_NS_ordered} from the vacuum vector $\Omega$, and with $L_{-1}$ playing the role of the infinitesimal translation operator $T$. We also have a conformal vector $\nu:=L_{-2}\Omega$ making $V$ a VOSA of CFT type with central charge $c$. Note also that $\nu$ together with the vector $\tau:=G_{-\frac{3}{2}}\Omega$ generate the whole VOSA. Moreover, they are Hermitian thanks to the symmetry of the fields $T(z)=z^2Y(\nu,z)$ and $G(z)=z^\frac{3}{2}Y(\tau,z)$, so that $V$ turns out to be unitary and simple by \cite[Proposition 3.29 and Propostion 3.10]{CGH} respectively. Then this is unitarily isomorphic to the $N=1$ super-Virasoro VOSA $V^c(NS)$, see \cite[Example 7.5 and Proposition 3.14]{CGH}.
\end{proof}

It follows an alternative proof to the one given in \cite{FQS85, GKO86}, \cite[Section 5]{KW86}:

\begin{cor}  \label{cor:unitarity_vacuum_representations_NS}
	The irreducible lowest weight representations of the Neveu-Schwarz algebra $NS$ with central charge $c\geq \frac{3}{2}$ and lowest weight $h=0$ are unitary.
\end{cor}

Turning back to the initial problem of constructing suitable unitary representations of $NS$, we implement the transformation $\varphi_{\kappa,\eta}=\varphi_{0,\eta}\circ\varphi_{\kappa,0}$ to \eqref{eq:bosonxfermion_representations_NS_fields_BS} with the same $\kappa\in\R$ and all $\eta\in\R$, so that we have:
\begin{equation}  \label{eq:bosonxfermion_representations_NS_fields_unitary_higher}
	\begin{split}
		T_{\kappa,\eta}(z)
		&:=
		\half:J^2:(z)+\half:\partial\Phi\Phi:(z)+\eta J(z)+\kappa J'(z)+\frac{\kappa^2+\eta^2}{2} \\
		G_{\kappa,\eta}(z)
		&:=
		:J\Phi:(z)+\eta\Phi(z)+2\kappa\Phi'(z) 
	\end{split}
\end{equation}   
where we have used that $\widetilde{\rho}$ satisfies the differential equation \eqref{eq:rho_differential_eq}.
Note that \eqref{eq:bosonxfermion_representations_NS_fields_unitary_higher} defines unitary representations of $NS$ as $T_{\kappa,\eta}(z)$ and $G_{\kappa,\eta}(z)$ are symmetric for all $\kappa,\eta\in\R$. 
Moreover, $\Omega$ is a lowest weight vector of lowest weight $h(\kappa,\eta):=\frac{\kappa^2+\eta^2}{2}$ for all $\kappa,\eta\in\R$. To sum up, we have construct an infinite number of unitary representations of $NS$ with central charge $c(\kappa)=\frac{3}{2}+12\kappa^2$ with $\kappa\in\R$ and lowest weight $h\geq \frac{c(k)-\frac{3}{2}}{24}$. 
Just for the sake of completeness, we also remark that we could have obtained \eqref{eq:bosonxfermion_representations_NS_fields_unitary_higher} directly from \eqref{eq:bosonxfermion_representations_NS_fields} implementing the transformations of type $\varphi_{0,i\kappa+\eta}$.

\subsubsection{Representations of the $N=1$ super-Virasoro nets}
\label{subsubsec:step2}

The goal is to translate the algebraic language of Subsection \ref{subsubsec:step1} into the operator algebraic one by adapting the approach pursued in \cite[Section 4]{BS90} and \cite[Section 4.4]{CTW22}, cf.\ also \cite[Section 4]{Car04}.

From now onward, $\kappa$ is any real number.
Let us call $V^\kappa$ the simple unitary VOSA obtained by Theorem \ref{theo:unitarity_BS_representation_NS} from fields \eqref{eq:bosonxfermion_representations_NS_fields_BS}. 
If $\widehat{Y}_\kappa$ is the corresponding state-field correspondence, $\widehat{\nu}^\kappa:=L_{\kappa,-2}\Omega$ and $\widehat{\tau}^\kappa:=G_{\kappa,-\frac{3}{2}}\Omega$, then $\widehat{Y}(\widehat{\nu}^\kappa,z)=z^{-2}T_\kappa(z)$ and $\widehat{Y}(\widehat{\tau}^\kappa,z)=z^{-\frac{3}{2}}G_\kappa(z)$.
Recall that $V$ is strongly graded-local and
$V^\kappa$ is too because it is isomorphic to the $N=1$ super-Virasoro VOSA with central charge $c(\kappa)=1+12\kappa^2$. 
Accordingly, let $\A_V$ and $\A_{V^\kappa}$ be the irreducible graded-local conformal nets arising by strong graded locality from the corresponding VOSAs on the Hilbert space completions $\mathcal{H}$ of $V$ and $\mathcal{H}^\kappa$ of $V^\kappa$ respectively. 
We highlight that every $V^\kappa$ is an inner product subspace of $V$, so that we can consider the orthogonal projection $E_\kappa$ from $\mathcal{H}$ to $\mathcal{H}^\kappa$.
Moreover, we denote the positive-energy strongly continuous projective unitary representations of $\Diff^+(S^1)^{(\infty)}$ of $\A_V$ and of $\A_{V^\kappa}$ by $U$ and $U_\kappa$ respectively.

We will use $C^\infty_c(\pSone)$ and $C^\infty_{\chi,c}(\pSone)$, with corresponding symbols for their subsets of real-valued functions, to denote the subspaces of $C^\infty(S^1)$ and $C^\infty_\chi(S^1)$ respectively of only functions with compact support in $\pSone$.
For all $f\in C^\infty_c(\pSone)$ and all $g\in C^\infty_{\chi, c}(\pSone)$, $\widetilde{\rho}f$ and $\widetilde{\rho}g$ are well-defined functions in $C^\infty_c(\pSone)$ and in $C^\infty_{\chi,c}(\pSone)$ respectively.
As in \eqref{eq:defin_smeared_vertex_operators}, we smeared the fields \eqref{eq:bosonxfermion_representations_NS_fields_BS}, defining $T_\kappa(f)$ and $G_\kappa(g)$ as the closures of the following respective operators on $V$:
\begin{equation}  \label{eq:defin_smeared_Tkappa_Gkappa}
	\begin{split}
		T_\kappa(f)_0 &:= Y_0(\nu,f)+i\kappa Y_0(\bos,-f')-\kappa Y_0(\bos,\widetilde{\rho}f) \\
		G_\kappa(g)_0 &:=  Y_0(\tau,g) +2i\kappa Y_0(\fer,-g') -\kappa Y_0(\fer, \widetilde{\rho}g) 
	\end{split}
\end{equation}
with $f'$ and $g'$ as in \eqref{eq:defin_derivative_test_functions}. 
In other words, for all $x\in V$,
\begin{equation}  \label{eq:defin_smeared_Tkappa_Gkappa_series}
	\begin{split}
		T_\kappa(f)_0x
		&:=
		\sum_{m\in\Z}\widehat{f}_m(L_m-i\kappa m J_m)x-\kappa\sum_{m\in\Z}\widehat{(\widetilde{\rho}f)}_mJ_mx\\
		G_\kappa(g)_0x
		&:=
		\sum_{n\in\Z-\half}\widehat{g}_n(G_n-2i\kappa n \Phi_n)x-\kappa\sum_{n\in\Z-\half}\widehat{(\widetilde{\rho}g)}_n\Phi_nx  
	\end{split}
\end{equation} 
where the right hand sides are well-defined thanks to the energy boundedness of $V$.

\begin{rem}
	An equivalent definition of the operators in \eqref{eq:defin_smeared_Tkappa_Gkappa} and \eqref{eq:defin_smeared_Tkappa_Gkappa_series} can be given in the \textit{real picture}, see \cite[Remark 5.1]{CGH} and references therein. Cf.\ \cite[Eq.\ (1.15)]{FST89}.
\end{rem}

\begin{lem}  \label{lem:properties_smeared_Tkappa_Gkappa}
	For all $f\in C^\infty_c(\pSone,\R)$ and all $g\in C^\infty_{\chi,c}(\pSone,\R)$, $T_\kappa(f)$ and $G_\kappa(g)$ are self-adjoint operators on $\mathcal{H}$.
	For all $f,\tilde{f}\in C^\infty_c(\pSone)$ and all $g,\widetilde{g}\in C^\infty_{\chi,c}(\pSone)$ such that $f$ and $g$ have supports disjoint from the ones of $\widetilde{f}$ and $\widetilde{g}$, 
	\begin{equation}
		W^*(\{T_\kappa(f), G_\kappa(g)\})
		\subseteq W^*(\{Y(\bos,\tilde{f}), ZY(\fer,\tilde{g})Z^*\})' \,.
	\end{equation}
	Furthermore, for all $f\in C^\infty_c(\pSone)$ and all $g\in C^\infty_{\chi,c}(\pSone)$,
	the orthogonal projection $E_\kappa$ is in $W^*(\{T_\kappa(f), G_\kappa(g)\})'$ and $E_\kappa T_\kappa(f)E_\kappa=\widehat{Y}(\widehat{\nu}^\kappa,f)$ and $E_\kappa G(g)E_\kappa =\widehat{Y}(\widehat{\tau}^\kappa, g)$.
\end{lem}

\begin{proof} 
If not differently specified, $f$ is in $C^\infty_c(\pSone)$ and $g$ is in $C^\infty_{\chi,c}(\pSone)$. 
Recall from Section \ref{section:preliminaries} that the subspace $\mathcal{H}^\infty$ of $\mathcal{H}$ of smooth vectors for $L_0$ is a common invariant core for all the smeared vertex operators on $V$.
It follows that $\mathcal{H}^\infty$ is also a common invariant core for the smeared fields $T_\kappa(f)$ and $G_\kappa(g)$. On $V$, $\Phi(z)$ satisfies $0$-th order energy bounds, $J(z)$ and $G(z)$ satisfy $\half$-th order energy bounds, whereas $T(z)$ satisfy linear energy bounds, see the proof of Theorem \ref{theo:linear_energy_bounds_strong_locality}. 
This implies that $T_\kappa(f)$ and $G_\kappa(g)$ satisfy linear energy bounds on $\mathcal{H}^\infty$, that is there exists non-negative real numbers $M_f$ and $M_g$ such that for all $c\in \mathcal{H}^\infty$,
\begin{equation}
			\norm{T_\kappa(f)c}
				\leq 
			M_f\norm{(L_0+1_{\mathcal{H}^\infty})c} 
			\quad\mbox{and}\quad
			\norm{G_\kappa(g)c}
			\leq 
			M_g\norm{(L_0+1_{\mathcal{H}^\infty}) c} \,.
\end{equation}
Moreover, it can be shown that also the commutators 
$$
[L_0,T_\kappa(f)]\,,\quad
[L_0,G_\kappa(g)]\,,\quad
[L_0,[L_0,T_\kappa(f)]]\,,\quad
[L_0,[L_0,G_\kappa(g)]]
$$
satisfy at most linear energy bounds in the above sense. Therefore, the first two statements follows at once from \cite[Theorem 19.4.4]{GJ87}, see also \cite[Theorem A.2]{AGT23}.

Working with the series expansion of $\widetilde{\rho}$ around $0$, it can be shown that for all $l\in\{f,g\}$
\begin{equation}
	\widehat{(\widetilde{\rho} l)}_m =i\widehat{l}_m+2i\sum_{j=1}^{+\infty}(-1)^j \widehat{l}_{m-j} \,.
\end{equation}
Then we have that $T_\kappa(f)$ and $G_\kappa(g)$ coincide with $\widehat{Y}(\widehat{\nu}^\kappa,f)$ and $\widehat{Y}(\widehat{\tau}^\kappa, g)$ on $V^\kappa$ respectively.
Set $\mathcal{D}^\kappa$ as the subspace of $\mathcal{H}^\infty$ generated by the polynomials 
\begin{equation}
T_\kappa(f_1)\cdots T_\kappa(f_n)G_\kappa(g_1)\cdots G_\kappa(g_m)\Omega
\end{equation}
for arbitrary finite collections of functions $\{f_i\}$ in $C^\infty_c(\pSone)$ and $\{g_j\}$ in $C^\infty_{\chi,c}(\pSone)$.
Then $\mathcal{D}^\kappa$ is a common invariant core for $T_\kappa(f)$ and $G_\kappa(g)$. It is not difficult to see that $\mathcal{D}^\kappa$ is also a core for both $\widehat{Y}(\widehat{\nu}^\kappa,f)$ and $\widehat{Y}(\widehat{\tau}^\kappa, g)$, see \cite[Lemma 7.2]{CKLW18}. Moreover, the latter operators coincide with $T_\kappa(f)$ and $G_\kappa(g)$ on $\mathcal{D}^\kappa$ respectively.

Set $A_f:=E_\kappa T_\kappa(f)E_\kappa$ and $B_g:=E_\kappa G(g)E_\kappa $ and consider them as operators on $\mathcal{H}^\kappa$.
Note that $A_f$ and $B_g$ are symmetric with domains given by $\mathcal{D}(A_f)=\mathcal{D}(T_\kappa(f))\cap \mathcal{H}^\kappa$ and $\mathcal{D}(B_g)=\mathcal{D}(G_\kappa(g))\cap \mathcal{H}^\kappa$ respectively.
Now, if $\Psi$ is in the domain of $\widehat{Y}(\widehat{\nu}^\kappa,f)$, then there exists a sequence $\{\Psi_n\}$ in $\mathcal{D}^\kappa$ converging to $\Psi$ and such that 
$$
\{T_\kappa(f)\Psi_n=\widehat{Y}(\widehat{\nu}^\kappa,f)\Psi_n\}
\to \widehat{Y}(\widehat{\nu}^\kappa,f)\Psi \,.
$$
Therefore, $\Psi\in \mathcal{D}(A_f)$ and $T_\kappa(f)\Psi =\widehat{Y}(\widehat{\nu}^\kappa,f)\Psi$. Moreover, 
\begin{equation}
	A_f=E_\kappa T_\kappa(f)\Psi
	=E_\kappa\widehat{Y}(\widehat{\nu}^\kappa,f)\Psi
	=\widehat{Y}(\widehat{\nu}^\kappa,f)\Psi
\end{equation}
which implies that $\widehat{Y}(\widehat{\nu}^\kappa,f)\subseteq A_f$. Let $f\in C^\infty_c(\pSone,\R)$. As $\widehat{Y}(\widehat{\nu}^\kappa,f)$ is self-adjoint ad thus maximally symmetric, then it must be equal to $A_f$. Moreover, we also know that if $\Psi\in \mathcal{D}(A_f)$, then $\Psi\in \mathcal{D}(T_\kappa(f))$ and thus
\begin{align}
	T_\kappa(f)\Psi &=\widehat{Y}(\widehat{\nu}^\kappa,f)\Psi \\
	E_\kappa T_\kappa(f)\Psi &=T_\kappa(f)\Psi \\
	A_f\Psi	&=T_\kappa(f)\Psi \,.
\end{align}
In other words, $T_\kappa(f) E_\kappa$ is equal to $E_\kappa T_\kappa(f)E_\kappa$ on $\mathcal{H}$. 
As $E_\kappa T_\kappa(f)E_\kappa$ is self-adjoint on $\mathcal{H}^\kappa$, then it is self-adjoint also on $\mathcal{H}$.
It follows that 
\begin{equation}
	T_\kappa(f)E_\kappa=(T_\kappa(f)E_\kappa)^*\supseteq E_\kappa T_\kappa(f)
\end{equation}
that is $E_\kappa\in W^*(T_\kappa(f))'$.
Then the result for generic $f\in C^\infty_c(\pSone)$ follows. 
We can conclude the proof noting that a similar argument works for $B_g$ too.
\end{proof}

Let $\J^\R$ be the subset of $\J$ of intervals whose closures do not contain the point $-1\in S^1$.
Set $\B^\kappa$ as the family of von Neumann algebras on $\mathcal{H}$ given by
\begin{equation}  \label{eq:defin_net_Bkappa}
	\B^\kappa(I):=W^*\left(\left\{ T_\kappa(f) \,,\,\, G_\kappa(g) \,\left|\,
	\begin{array}{l}
		f\in C^\infty_c(\pSone)\,,\,\,\mathrm{supp}f\subset I \\
		g\in C^\infty_{\chi,c}(\pSone) \,,\,\,\mathrm{supp}g\subset I
	\end{array} 
	\right.\right\}\right) 
\end{equation}
for all $I\in\J^\R$.
By Lemma \ref{lem:properties_smeared_Tkappa_Gkappa}, for all $I\in\J^\R$, $\B^\kappa(I)\subseteq \A_V(I)$ and there is an isomorphism $p^\kappa_I:\B^\kappa(I)\to\A_{V^\kappa}(I)$, giving rise to a soliton of $\A_{V^\kappa}$ on $\mathcal{H}$. 

From now onward, $\eta$ is any real number.
From \eqref{eq:bosonxfermion_representations_NS_fields_unitary_higher}, define the smeared fields $T_{\kappa,\eta}(f)$ and $G_{\kappa,\eta}(g)$ as in \eqref{eq:defin_smeared_vertex_operators}, allowing all $f\in C^\infty(S^1)$ and all $g\in C^\infty_\chi(S^1)$, which is possible as no $\widetilde{\rho}(z)$ appears in those formulae.
Moreover, thanks to the unitarity of the $NS$ representation that they generate, $T_{\kappa,\eta}(f)$ and $G_{\kappa,\eta}(g)$ are self-adjoint whenever $f$ and $g$ are chosen to be real-valued. 
Reasoning as in the first paragraph of the proof of Lemma \ref{lem:properties_smeared_Tkappa_Gkappa}, we can deduce that $T_{\kappa,\eta}(f)$ and $G_{\kappa,\eta}(g)$ satisfy linear energy bounds in the sense given there. It follows that for all $f,\tilde{f}\in C^\infty(S^1)$ and all $g,\widetilde{g}\in C^\infty_\chi(S^1)$ such that $f$ and $g$ have supports disjoint from the ones of $\widetilde{f}$ and $\widetilde{g}$, $W^*(\{T_{\kappa,\eta}(f), G_{\kappa,\eta}(g)\})$ is contained in $W^*(\{Y(\bos,\tilde{f}), ZY(\fer,\tilde{g})Z^*\})'$.
Denote the Fourier coefficients of $T_{\kappa,\eta}(z)$ and $G_{\kappa,\eta}(z)$ by $\{L_{\kappa,\eta,m}\mid m\in\Z\}$ and $\{G_{\kappa,\eta,n}\mid n\in\Z-\half\}$ respectively.  
Let $\mathcal{H}^{\kappa,\eta}\subseteq\mathcal{H}$ be the completions of the vector subspace $V^{\kappa,\eta}$ of $V$ arising from those Fourier coefficients applied to the vacuum vector $\Omega$. 
Note that $L_{\kappa,\eta,0}=L_0+\frac{\kappa^2+\eta^2}{2}$.
It follows that the orthogonal projection $E_{\kappa,\eta}$ from $\mathcal{H}$ to $\mathcal{H}^{\kappa,\eta}$ restricts to an orthogonal projection from $V$ to $V^{\kappa,\eta}$.
As a consequence, $E_{\kappa,\eta}$ is in $W^*(\{T_{\kappa,\eta}(f), G_{\kappa,\eta}(g)\})'$.
Set the operators $\widehat{T_{\kappa,\eta}}(f)$ and $\widehat{G_{\kappa,\eta}}(g)$ on $\mathcal{H}^{\kappa,\eta}$ as the closures of the operators $T_{\kappa,\eta}(f)\restriction_{V^{\kappa,\eta}}$ and $G_{\kappa,\eta}(g)\restriction_{V^{\kappa,\eta}}$ respectively.
Similarly to what done in \eqref{eq:defin_net_Bkappa}, we can define $\B^{\kappa,\eta}$ as the family of von Neumann algebras on $\mathcal{H}$
\begin{equation}  
	\B^{\kappa,\eta}(I):=W^*\left(\left\{ T_{\kappa,\eta}(f) \,,\,\, G_{\kappa,\eta}(g) \,\left|\,
	\begin{array}{l}
		f\in C^\infty(S^1)\,,\,\,\mathrm{supp}f\subset I \\
		g\in C^\infty_\chi(S^1) \,,\,\,\mathrm{supp}g\subset I
	\end{array} 
	\right.\right\}\right) 
\end{equation}
for all $I\in\J$.
Then for all $I\in\J$, we have that $\B^{\kappa,\eta}(I)\subseteq \A_V(I)$. Moreover, let $\A^{\kappa,\eta}$ be the family of von Neumann algebras on $\mathcal{H}^{\kappa,\eta}$
\begin{equation}  
	\A^{\kappa,\eta}(I):=W^*\left(\left\{ \widehat{T_{\kappa,\eta}}(f) \,,\,\, \widehat{G_{\kappa,\eta}}(g) \,\left|\,
	\begin{array}{l}
		f\in C^\infty(S^1)\,,\,\,\mathrm{supp}f\subset I \\
		g\in C^\infty_\chi(S^1) \,,\,\,\mathrm{supp}g\subset I
	\end{array} 
	\right.\right\}\right) 
\end{equation}
for all $I\in\J$. Then for all $I\in\J$, there is an isomorphism $p^{\kappa,\eta}_I:\B^{\kappa,\eta}(I)\to \A^{\kappa,\eta}(I)$. 
We denote by $U_{\kappa,\eta}$ the positive-energy strongly continuous projective unitary representation of $\Diff^+(S^1)^{(\infty)}$ on $\mathcal{H}^{\kappa,\eta}$ obtained integrating the representation of the Virasoro algebra given by $\{L_{\kappa,\eta,m}\mid m\in\Z\}$, see \cite{Tol99}, see also \cite[Theorem 3.4]{CKLW18} and references therein.
Following the argument used in \cite[pp.\ 1100--1103]{CKL08}, see also \cite[Proposition 6.4]{CKLW18} and \cite[Proposition 4.24 and Remark 4.25]{CGH}, it can be proved that the family $\A^{\kappa,\eta}$ is diffeomorphism covariant with respect to $U_{\kappa,\eta}$.

Now, the strategy is to implement the transformation $\varphi_{\kappa,\eta}$ introduced in Subsection \ref{subsubsec:step1} with an argument as for \textit{BMT representations}\label{bmt_representations}, see e.g.\ \cite[Section III]{BMT88}, \cite[Section 2 and Section 3]{BMT90}, \cite[Section 3.1]{GLW98}, \cite[Section 4]{Car03}, \cite[Section 4]{Car04}. 
One can define a soliton $\alpha^{\kappa,\eta}$ of $\A_V$ such that for all $I\in\J^\R$, all $f\in C^\infty_c(\pSone,\R)$ and all $g\in C^\infty_{\chi,c}(\pSone,\R)$ with supports contained in $I$,
\begin{align}
		\alpha^{\kappa,\eta}_I(e^{iY(\bos,f)})
			&=
		W(u_{\kappa,\eta,I})e^{iY(\bos,f)}W(u_{\kappa,\eta,I})^*
		=e^{i\widetilde{\rho}_{\kappa,\eta}[f]} e^{iY(\bos,f)} \\
		\alpha^{\kappa,\eta}_I(e^{iY(\fer,g)})
			&= e^{iY(\fer,g)} 
\end{align}
where
\begin{equation}
 	\widetilde{\rho}_{\kappa,\eta}[f]:=\oint_{S^1}(\kappa\widetilde{\rho}(z)+\eta)f(z)\frac{\mathrm{d}z}{2\pi iz} 
\end{equation}
and the \textit{Weyl operator} $W(u_{\kappa,\eta,I}):=e^{iY(\bos,u_{\kappa,\eta,I})}$ is given by a suitable function $u_{\kappa,\eta,I}\in C^\infty_c(\pSone)$ with $\mathrm{supp}u_{\kappa,\eta,I}$ contained in some $J\in\J^\R$ satisfying $\overline{I}\subseteq J$. 
Composing for all $I\in\J^\R$, the restriction of $\alpha^{\kappa,\eta}_I$ to $\B^\kappa(I)$ with $(p^{\kappa}_I)^{-1}$ and $p^{\kappa,\eta}_I$ defined above, we obtain a soliton $\pi^{\kappa,\eta}$ of $\A_{V^\kappa}$ on $\mathcal{H}^{\kappa,\eta}$, that is a family of normal representations
\begin{equation}  \label{eq:soliton_kappa_eta}
	\pi^{\kappa,\eta}_I:
	\A_{V^\kappa}(I)\to \A^{\kappa,\eta}(I)\subseteq B(\mathcal{H}^{\kappa,\eta})
	\,,\quad I\in\J^\R
\end{equation}
such that for all $I\in\J^\R$, all $f\in C^\infty_c(\pSone,\R)$ and all $g\in C^\infty_{\chi,c}(\pSone,\R)$ with supports contained in $I$,
\begin{equation}
		\pi^{\kappa,\eta}_I(e^{i\widehat{Y}(\widehat{\nu}^\kappa,f)})
			= e^{i\widehat{T_{\kappa,\eta}}(f)} 
			\quad\mbox{and}\quad
		\pi^{\kappa,\eta}_I(e^{i\widehat{Y}(\widehat{\tau}^\kappa,g)})
		= e^{i\widehat{G_{\kappa,\eta}}(g)} 
\end{equation}
Furthermore, we also have that for all $I\in\J^\R$, all $\gamma\in \Diff^+(S^1)^{(\infty)}$ such that $\dot{\gamma}I\in \J^\R$ and all $X\in \A_{V^\kappa}(I)$,
\begin{equation} \label{eq:covariance_soliton_kappa_eta}
	 U_{\kappa,\eta}(\gamma) \pi^{\kappa,\eta}_I(X) U_{\kappa,\eta}(\gamma)^*
	 =
	\pi^{\kappa,\eta}_{\dot{\gamma}I}(U_\kappa(\gamma) XU_\kappa(\gamma)^*)
\end{equation}
thanks to the covariance properties of $\A_{V^\kappa}$ and of $\A^{\kappa,\eta}$.

Then we need the following proposition:

\begin{prop} \label{prop:extending_solitons}
	Let $\A$ be an irreducible graded-local conformal net on $S^1$ and let $\pi$ be a soliton of $\A$. Suppose that there exists a representation $U_\pi$ of $\Mob(S^1)^{(\infty)}$ (resp.\ $\Diff^+(S^1)^{(\infty)}$) such that $U_\pi(r^{(\infty)}(2\pi))$ implements the grading. Suppose further that for all $\gamma$ $\in$ $\Mob(S^1)^{(\infty)}$ (resp.\ $\Diff^+(S^1)^{(\infty)}$) such that $\dot{\gamma}I\in\J^\R$,
	\begin{equation} \label{eq:covariance_soliton}
		U_\pi(\gamma) \pi_I(X) U_\pi(\gamma)^*
		=
		\pi_{\dot{\gamma}I}(U(\gamma) XU(\gamma)^*)
		\,.
	\end{equation}
	Then $\pi$ extends to a M\"{o}bius (resp.\ diffeomorphism) covariant representation of $\A$.
\end{prop}

\begin{proof} 
	This is a straightforward adaptation of the proof of \cite[Lemma 2.6]{CHKLX15}. 
	Let $I\in\J$, then there exists $t\in\R$ such that $r(t)I\in\J^\R$. For all $X\in \A(I)$, we define
	\begin{equation}  \label{eq:representation_extension}
		\pi_I(X):= U_\pi(r^{(\infty)}(-t))\pi_{r(t)I}(U(r^{(\infty)}(t))XU(r^{(\infty)}(-t))) U_\pi(r^{(\infty)}(t)) \,.
	\end{equation}
	Using the fact that $U_\pi$ implements the grading, it is not hard to verify that \eqref{eq:representation_extension} is well-defined, that is it does not depend on the choice of $t\in\R$. 
	As far as the covariance is concerned, let $\gamma\in \Mob(S^1)^{(\infty)}$ (resp.\ $\Diff^+(S^1)^{(\infty)}$) and $I\in \J$. Pick $t,s\in\R$ such that $r(t)I\in\J^\R$ and that $r(s)\dot{\gamma} I\in\J^\R$.
	Then we have that:
	(note that to lighten the notation, we will not distinguish between $\gamma, r^{(\infty)}(t), r^{(\infty)}(s) \in \Mob(S^1)^{(\infty)}$ (resp.\ $\Diff^+(S^1)^{(\infty)}$) and their corresponding projections onto $\dot{\gamma}, r(t), r(s) \in \Mob(S^1)$ (resp.\ $\Diff^+(S^1)$), leaving the correct interpretation of the symbols to the reader)
	\begin{equation}
		\begin{split}
			U_\pi(\gamma)\pi_I(X)U_\pi(\gamma)^*
				&=
			U_\pi(\gamma r(-t))\pi_{r(t)I}(U(r(t))XU(r(-t)))U_\pi(\gamma r(-t))^* \\
				&=
				\begin{split}
					U_\pi(r(-s))U_\pi(r(s)\gamma r(-t))
					\pi_{r(t)I} 
						&
					(U(r(t))XU(r(-t))) \cdot \\
			     		&
			     	\cdot U_\pi(r(s)\gamma r(-t))^*U_\pi(r(s))
				\end{split} \\
				&=
			U_\pi(r(-s))\pi_{r(s)\gamma I} 
			(U(r(s)\gamma)XU(r(s)\gamma)^*) U_\pi(r(s)) \\
				&=
			\pi_{\gamma I} 
			(U(\gamma)XU(\gamma)^*) 
		\end{split}
	\end{equation}
	where we have used the covariance \eqref{eq:covariance_soliton} for the third equality and definition \eqref{eq:representation_extension} for the first and last equality.
	This shows that \eqref{eq:representation_extension} gives rise to the desired M\"{o}bius (resp.\ diffeomorphism) covariant representation of $\A$ extending $\pi$.
\end{proof}

Proposition \ref{prop:extending_solitons} applies to the soliton $\pi^{\kappa,\eta}$ defined by \eqref{eq:soliton_kappa_eta} thanks to the covariance \eqref{eq:covariance_soliton_kappa_eta}.
Standard arguments show that, as a consequence of the irreduciblity of the underlying representation of the Neveu-Schwarz algebra $NS$, $\pi^{\kappa,\eta}$ is an irreducible representation of $\A_{V^\kappa}$.
In conclusion, letting $\kappa$ and $\eta$ vary in $\R$, it follows that:

\begin{theo}  \label{theo:irreprs_N=1_net}
	Let $c\geq \frac{3}{2}$. Then for all real numbers $h\geq \frac{c-\frac{3}{2}}{24}$, there exists an irreducible representation $\pi^h$ of the $N=1$ super-Virasoro net $\SVirnet_c$ of lowest energy $h$. 
\end{theo}

\subsection{\texorpdfstring{$N=2$ super-Virasoro nets}{N=2 super-Virasoro nets}}  \label{subsec:N=2_super-virasoro_nets}
The simple minimal $W$-algebras $W_k(\g,f)$ with $\g=\mathfrak{spo}(2|2)$ are the (simple) $N=2$ \textit{super-Virasoro VOSAs} $V^{c(k)}(N2)$ with central charges $c(k)$ as in the left hand side of \eqref{eq:unitary_series_N=2}, see \cite[Section 7]{KRW03} and \cite[Section 8.3]{KW04}.
These VOSAs are usually obtained through representations of the $N=2$ \textit{Lie superalgebra} $N2$, see \cite[Eq.s (5.9.4a) and (5.9.4b)]{Kac01}.
The \textit{unitary series} for $V^{c(k)}(N2)$ is as follows:
\begin{equation}  \label{eq:unitary_series_N=2}
	c(k)=-3(2k+1)
	\left\{\begin{array}{ll}
		=3\left(1-\frac{2}{p}\right) \,, &  k=\frac{1}{p}-1\,,\,\, p\in \Z_{\geq 2} \\
		=3 \,,&  k=-1 \\
		> 3 \,,& k<-1
	\end{array}	\right. 
\end{equation}
Moreover, $V^{c(k)}(N2)$ turns out to be strongly rational for $c(k)$ in the \textit{unitary discrete series}, that is the first row in \eqref{eq:unitary_series_N=2}, whereas it is not rational for $c(k)\geq 3$, see \cite{BFK86, DPYZ86, Ioh10}, cf.\ also \cite[Theorem 3.2]{CHKLX15}.

We use $\mathrm{SVir2}_{c(k)}$ to denote the $N=2$ \textit{super-Virasoro net} $\A_{V^{c(k)}(N2)}$, see \cite[Example 7.17]{CGH} and references therein. As in the case of the Neveu-Schwarz Lie superlagebra $NS$ in Section \ref{subsec:N=1_super-virasoro_nets}, to the best of our knowledge, it is still unproved that $\A_{V^{c(k)}(N2)}$ is not completely rational for central charges $c(k)\geq3$, whereas the complete rationality for all cases with $c(k)<3$ is actually proved, using the coset construction of these models, see \cite[Section 5]{CHKLX15}.
Also in this case, we remark that a complete classification of irreducible graded-local conformal net extensions of $\operatorname{SVir2}_{c(k)}$ with $c(k)$ in the unitary discrete series was obtained in \cite[Section 6]{CHKLX15} and see also \cite[Theorem 6.3]{CGGH23} for the corresponding statement in the VOSA setting.

For the construction of representations of $\mathrm{SVir2}_c$ with $c\geq3$, consider the graded tensor product $V:=M(1)^{\hat{\otimes}2}\hat{\otimes} F^{\hat{\otimes}2}$ of two copies of the Heisenberg VOA with two copies of the real free fermion VOSA, see references in Remark \ref{rem:unitary_range_and_collapsing} and Example \ref{ex:real_free_fermion} respectively. 
The unitary structure on $V$ is the one induced by the ones of $M(1)$ and $F$ as explained in Section \ref{section:preliminaries}.
Let $\bos^\pm$ and $\fer^\pm$ be four Hermitian primary generators with norm one for the two copies of $M(1)$ and the two of $F$ in $V$ of conformal weights $1$ and $\half$ respectively. 
This implies that their images in $V$, which we denote with the same symbols, are still Hermitian and primary.
By \cite[Example 5.9d]{Kac01}, where $\alpha^\pm$ and $\varphi^\pm$ there are chosen as $\frac{1}{\sqrt{2}}(\bos^+\pm i\bos^-)$ and $\frac{1}{\sqrt{2}}(\fer^+\pm i\fer^-)$ here respectively,
there is a family of parameter $\varsigma\in \C$ in $V$ of vectors giving rise to a family of $N=2$ super-Virasoro algebras with central charges $c(\varsigma)=3+6\varsigma^2$:
\begin{equation} \label{eq:bosonsxfermions_representations_N2_pm} 
	\begin{split}
		\nu^\varsigma
			&:=
		\frac{1}{2}\left( (\bos^+_{-1})^2 + (\bos^-_{-1})^2 + \fer^+_{-\frac{3}{2}}\fer^+_{-\half} + \fer^-_{-\frac{3}{2}}\fer^-_{-\half} -i\sqrt{2}\varsigma \bos^-_{-2} \right)\Omega \\
		\tau^{+,\varsigma}
			&:=
		\half\left(\bos^+_{-1}\fer^+_{-\half} - \bos^-_{-1}\fer^-_{-\half}+\sqrt{2}\varsigma \fer^+_{-\frac{3}{2}} \right)\Omega
		+\frac{i}{2}\left(\bos^+_{-1}\fer^-_{-\half} + \bos^-_{-1}\fer^+_{-\half}+\sqrt{2}\varsigma \fer^-_{-\frac{3}{2}}\right)\Omega \\
		\tau^{-,\varsigma}
			&:=
		\half\left(\bos^+_{-1}\fer^+_{-\half} - \bos^-_{-1}\fer^-_{-\half}-\sqrt{2}\varsigma \fer^+_{-\frac{3}{2}} \right)\Omega
		-\frac{i}{2}\left(\bos^+_{-1}\fer^-_{-\half} + \bos^-_{-1}\fer^+_{-\half}-\sqrt{2}\varsigma \fer^-_{-\frac{3}{2}}\right)\Omega \\
		j^\varsigma
			&:=
		\left(-i\fer^+_{-\half}\fer^-_{-\half}-\varsigma\sqrt{2}\bos^+_{-1}\right) \Omega \,.
		\end{split}
\end{equation}
For all $\varsigma\in\C$, set $\tau^{1,\varsigma}:=\frac{1}{\sqrt{2}}(\tau^{+,\varsigma}+\tau^{-,\varsigma})$ and $\tau^{2,\varsigma}:=\frac{-i}{\sqrt{2}}(\tau^{+,\varsigma}-\tau^{-,\varsigma})$.
Then we consider the following family:
\begin{equation} \label{eq:bosonsxfermions_representations_N2_12} 
	\begin{split}
		\nu^\varsigma
		&:=
		\frac{1}{2}\left( (\bos^+_{-1})^2 + (\bos^-_{-1})^2 + \fer^+_{-\frac{3}{2}}\fer^+_{-\half} + \fer^-_{-\frac{3}{2}}\fer^-_{-\half} -i\sqrt{2}\varsigma \bos^-_{-2} \right)\Omega \\
		\tau^{1,\varsigma}
		&:=
		\left(\frac{\bos^+_{-1}\fer^+_{-\half} - \bos^-_{-1}\fer^-_{-\half}}{\sqrt{2}}+i\varsigma \fer^-_{-\frac{3}{2}} \right)\Omega \\
		\tau^{2,\varsigma}
		&:=
		\left(\frac{\bos^+_{-1}\fer^-_{-\half} + \bos^-_{-1}\fer^+_{-\half}}{\sqrt{2}}-i\varsigma \fer^+_{-\frac{3}{2}}\right)\Omega \\
		j^\varsigma
		&:=
		\left(-i\fer^+_{-\half}\fer^-_{-\half}-\varsigma\sqrt{2}\bos^+_{-1}\right) \Omega \,.
	\end{split}
\end{equation}
Note that for $\varsigma=0$, we obtain a realization of the unitary $N=2$ super-Virasoro VOSA $V^3(N2)$ of central charge $c(0)=3$ as unitary subalgebra of $V$. Accordingly, set $\nu:=\nu^0$, that is the conformal vector of $V$, $\tau^1:=\tau^{1,0}$, $\tau^2:=\tau^{2,0}$ and $j:=j^0$.
We also use the standard notation 
\begin{equation}  \label{eq:N2_vertex_operators}
	\begin{split}
		Y(\nu,z)=\sum_{m\in\Z}L_mz^{-m-2} \,,&\quad
		Y(j,z)=\sum_{n\in\Z}J_nz^{-n-1} \\
		Y(\tau^1,z)=\sum_{r\in\Z-\half}G^1_rz^{-r-\frac{3}{2}} \,,&\quad
		Y(\tau^2,z)=\sum_{s\in\Z-\half}G^2_sz^{-s-\frac{3}{2}} \,.
	\end{split}
\end{equation}
We recall for further use the \textit{$N2$ commutation relations} with a general central charge $c\in\C$, satisfied by the Fourier coefficients in \eqref{eq:N2_vertex_operators} with $c=3$, see e.g.\ \cite[Definition 3.1]{CHKLX15}: 
\begin{equation}  \label{eq:N2_cr}
	\begin{split}
		\forall m,n\in \Z \qquad
		[L_m, L_n]
		&=
		(m-n)L_{m+n}+\frac{c(m^3-m)}{12}\delta_{m,-n} 
		\\
		\forall l\in\{1,2\} \,\,\,\forall m\in\Z\,\,\,\forall r\in\Z-\half \qquad
		[L_m, G^l_r]
		&= 
		\left(\frac{m}{2}-r\right) G^l_{m+r} 
		\\
		\forall l\in\{1,2\} \,\,\,\forall r,s\in\Z-\half \qquad
		[G^l_r, G^l_s]
		&=
		2L_{r+s}+\frac{c}{3}\left(r^2-\frac{1}{4}\right)\delta_{r,-s} 
		\\
		\forall r,s\in\Z-\half\qquad
		[G^1_r, G^2_s]
		&=
		i(r-s)J_{r+s} 
		\\
		\forall r\in\Z-\half \,\,\, \forall n\in\Z \qquad
		[G^1_r, J_n]
		&=
		- i G^2_{r+n}
		\\
		\forall r\in\Z-\half \,\,\, \forall n\in\Z \qquad
		[G^2_r, J_n]
		&=
		i G^1_{r+n}
		\\
		\forall m,n\in\Z\qquad
		[L_m, J_n]
			&=
		-nJ_{m+n}
		\\
		\forall m,n\in \Z  \qquad
		[J_m,J_n]
			&=
		\frac{c}{3}m \delta_{m,-n} 
		\,.
	\end{split}
\end{equation}
Note that for all $\varsigma\in\C$, $\{\nu^\varsigma,\tau^{1,\varsigma}\}$ as well as $\{\nu^\varsigma,\tau^{2,\varsigma}\}$ satisfies the $NS$ commutation relations \eqref{eq:NS_cr}.
Moreover, $\nu$, $\tau^1$, $\tau^2$ and $j^\varsigma$ for all $\varsigma\in\R$ are Hermitian vectors.

It should be clear that choosing $\varsigma=-\sqrt{2}\kappa$ with $\kappa\in\R$ and thus $c(\kappa):=c(\varsigma)=3+12\kappa^2$, we have a similar situation to the one in Section \ref{subsec:N=1_super-virasoro_nets}, so that we can try to proceed similarly to get the non-complete rationality of $\mathrm{SVir2}_c$ for all $c\geq 3$. 
Accordingly, we use the formalism introduced in Subsection \ref{subsubsec:formal_series_fields} and we refer to Subsection \ref{subsubsec:step1} and Subsection \ref{subsubsec:step2} for unexplained notations and details.
Set $J^\pm(z):=zY(\bos^\pm,z)$, $\Phi^\pm(z):=z^\half Y(\fer^\pm,z)$ and their shifted normal powers similarly to what has been done in \eqref{eq:shifted_normal_powers}.
Define also 
\begin{equation}
	\begin{split}
		T(z)
			&:=
		z^2Y(\nu,z)= \half:(J^+)^2:(z)+\half:(J^-)^2:(z)+\half:\partial\Phi^+\Phi^+:(z)+\half:\partial\Phi^-\Phi^-:(z) \\
		G^1(z)
			&:=
		z^\frac{3}{2}Y(\tau^1,z) = \frac{1}{\sqrt{2}}:J^+\Phi^+:(z)-\frac{1}{\sqrt{2}}:J^-\Phi^-:(z)\\
		G^2(z)
			&:=
		z^\frac{3}{2}Y(\tau^2,z) = \frac{1}{\sqrt{2}}:J^+\Phi^-:(z)+\frac{1}{\sqrt{2}}:J^-\Phi^+:(z)\\
		 J(z)
		 	&:=
		 zY(j,z)=-i:\Phi^+\Phi^-:(z) \,.
	\end{split}
\end{equation}
Then we rewrite \eqref{eq:bosonsxfermions_representations_N2_12} in terms of fields:
\begin{equation}  \label{eq:bosonsxfermions_representations_N2_fields}
	\begin{split}
		\widetilde{T}(z;\kappa)
			&:=
		z^2Y(\nu^\varsigma,z) 
		=T(z)-i\kappa\left(J^-(z)+i(J^-)'(z)\right) \\
		\widetilde{G^1}(z;\kappa)
			&:=
		z^\frac{3}{2}Y(\tau^{1,\varsigma},z)
		= G^1(z)+\frac{i\kappa}{\sqrt{2}}\left(\Phi^-(z)+2i(\Phi^-)'(z)\right)\\
		\widetilde{G^2}(z;\kappa)
		&:=
		z^\frac{3}{2}(\tau^{2,\varsigma},z)
		= G^2(z)-\frac{i\kappa}{\sqrt{2}}\left(\Phi^+(z)+2i(\Phi^+)'(z)\right)\\
		\widetilde{J}(z;\kappa)
			&:=
		zY(j,z)
		=J(z)+2\kappa J^+(z) \,.
	\end{split}
\end{equation}
The Fourier coefficients of  \eqref{eq:bosonsxfermions_representations_N2_fields}, namely 
$$
\left\{\left.\widetilde{L}_{\kappa,m}\,,\,\,\widetilde{G^1}_{\kappa,r}\,,\,\,\widetilde{G^2}_{\kappa,s}\,,\,\,\widetilde{J}_{\kappa,n}\,\right|\, m,n\in\Z\,,\,\, r,s\in\Z-\half\right\}
$$ 
induce lowest weight representations of the $N=2$ super-Virasoro algebra $N2$ on $V$ with central charges $c(\kappa)=3+12\kappa^2$ and lowest weight vector $\Omega$ of lowest weight $(h,q)=(0,0)$, where $\widetilde{L}_{\kappa,0}\Omega=h\Omega$ and $\widetilde{J}_{\kappa,0}\Omega=q\Omega$.
Note that unless $\kappa=0$ in \eqref{eq:bosonsxfermions_representations_N2_fields}, only $\widetilde{J}(z;\kappa)$ is symmetric.
This implies that for all non-zero $\kappa\in\R$, this representation  is not unitary with respect to the scalar product $\scalar$ on $V$.
Therefore, we have to perform a transformation of kind $\varphi_{-\kappa,i\kappa}$ applied to $J^-(z)$ to get a $N2$ version of \cite[Eq.\ (4.6)]{BS90} as we have done for the $NS$ case obtaining \eqref{eq:bosonxfermion_representations_NS_fields_BS}. 
What we get is:
\begin{equation}  \label{eq:bosonsxfermions_representations_N2_BS}
	\begin{split}
		T_\kappa(z)
			&:= 
		T(z)+\kappa \left((J^-)'(z)-\widetilde{\rho}(z)J^-(z)\right) \\ 
		G^1_\kappa(z)
			&:=
		G^1(z)-\frac{\kappa}{\sqrt{2}}\left(2(\Phi^-)'(z)-\widetilde{\rho}(z)\Phi^-(z) \right)\\
		G^2_\kappa(z)
			&:=
		G^2(z)+\frac{\kappa}{\sqrt{2}}\left(2(\Phi^+)'(z)-\widetilde{\rho}(z)\Phi^+(z) \right)\\
		J_\kappa(z)
			&:=
		J(z)+2\kappa J^+(z) \,.
	\end{split}
\end{equation}
Again, their Fourier coefficients:
\begin{equation}  \label{eq:bosonsxfermions_representations_N2_fields_BS_fourier}
	\begin{split}
		L_{\kappa, m}
			&= 
		L_m-i\kappa\left[(1+m)J^-_m+2\sum_{j<0}(-1)^jJ^-_{m-j}\right]
		,\quad m\in\Z\\
		G^1_{\kappa,r}
			&=  
		G^1_r +i\frac{\kappa}{\sqrt{2}}\left[(1+2r)\Phi^-_r+2\sum_{j<0}(-1)^j\Phi^-_{r-j}\right]
		,\quad  r\in\Z-\half \\
		G^2_{\kappa,s}
			&=  
		G^2_s -i\frac{\kappa}{\sqrt{2}}\left[(1+2s)\Phi^+_s+2\sum_{j<0}(-1)^j\Phi^+_{s-j}\right]
		,\quad s\in\Z-\half \\
		J_{\kappa,n}
			&=
		J_n+2\kappa J^+_n \,,\quad  n\in\Z 
	\end{split}
\end{equation}
give rise to lowest weight representations of the $N=2$ super-Virasoro algebra $N2$ on $V$ with central charges $c(\kappa)=3+12\kappa^2$ and lowest weight vector $\Omega$ of lowest weight $(h,q)=(0,0)$.
It is easy to see that the fields in \eqref{eq:bosonsxfermions_representations_N2_BS} satisfy a weak symmetry as the one in \eqref{eq:weak_symmetry_NS}. Actually, we note that $J_\kappa(z)$ is also symmetric. Therefore, we can state the following important result:

\begin{theo} \label{theo:unitarity_BS_representation_N2}
	Let $\left\{L_m \,,\,\, G^1_r \,,\,\, G^2_s \,,\,\, J_n\mid m,n\in\Z \,,\,\, r,s\in\Z-\half\right\}$ be a representation of the $N=2$ super-Virasoro algebra $N2$ on a vector superspace $V$ with central charge $c\geq 3$ and such that: 
	\begin{itemize}
		\item [(i)] there exists a cyclic lowest weight vector $\Omega$ such that $L_0\Omega=0$ and $J_0\Omega=0$;
		
		\item [(ii)] there exists a scalar product $\scalar$ with respect to $J(z)=\sum_{n\in\Z}J_nz^{-n}$ is symmetric, that is $J_n^\dagger=J_n$ for all $n\in\Z$, and such that the fields $T(z)=\sum_{m\in\Z} L_mz^{-m}$ and $G^i(z)=\sum_{r\in\Z-\half} G^i_rz^{-s}$ for all $i\in\{1,2\}$ satisfy
		\begin{equation}  \label{eq:weak_symmetry_N2_theo}
				\left(p(z)T(z)\right)^\dagger=\overline{p(z)}T_\kappa(z)
				\quad\mbox{ and }\quad
				\left(p(z)z^\half G^i(z)\right)^\dagger=\overline{p(z)}z^{-\half}G^i(z)
		\end{equation}
		for all trigonometric polynomials $p(z)$ such that $p(-1)=0$. 
	\end{itemize}
	Then $L_m^\dagger=L_{-m}$ and $(G^i_r)^\dagger=G^i_{-r}$ for all $i\in\{1,2\}$, for all $m\in\Z$ and all $r\in\Z-\half$, that is the representation is unitary.
	
	Furthermore, this representation has a structure of simple unitary VOSA with central charge $c$ (unitarily) isomorphic to the $N=2$ super-Virasoro VOSA $V^c(N2)$.
\end{theo}

\begin{proof}
	The proof relays on the one of Theorem \ref{theo:unitarity_BS_representation_NS}, which we will refer to for details, mostly thanks to the fact that the two sets 
	$$
	\left\{L_m \,,\,\, G^1_r \mid m\in\Z \,,\,\, r\in\Z-\half\right\}
	 \qquad\mbox{and}\qquad 
	\left\{L_m \,,\,\, G^2_s \mid m\in\Z \,,\,\, s\in\Z-\half\right\}
	$$ 
	are two representations of a Neveu-Schwarz algebra $NS$.
	Therefore, we can prove in the same way as there that $L_{-1}\Omega=G^1_{-\half}\Omega=G^2_{-\half}\Omega=0$.
	
	The goal is to prove that $L_0$ is unitary with respect to $\scalar$. We sketch the proof here below, referring to the proof of Theorem \ref{theo:unitarity_BS_representation_NS} for omitted details.
	The representation of $N2$ is generated by vectors of the following type:
	\begin{equation}  \label{eq:generic_generator_N2}
		\begin{split}
		\Psi&:=L_{-m_1}\cdots L_{-m_u}G^1_{-r_1}\cdots G^1_{-r_t}G^2_{-s_1}\cdots G^2_{-s_x}J_{-n_1}\cdots J_{-n_v}\Omega \\
		\Psi'&:=L_{-m_1'}\cdots L_{-m_{u'}'}G^1_{-r_1'}\cdots G^1_{-r_{t'}'}G^2_{-s_1'}\cdots G^2_{-s_{x'}'}J_{-n_1'}\cdots J_{-n'_{v'}}\Omega
		\end{split}
	\end{equation}
	for $u,t,x,v\in\Zpluseq$, for positive integers $m_1,\dots, m_u$ and $n_1,\dots, n_v$, for positive semi-integers $r_1,\dots, r_t$ and $s_1,\dots, s_x$ and with same conditions for their primed counterparts. 
	Suppose that $L_0\Psi=d\Psi$, $L_0\Psi'=d'\Psi'$ and that $d\not=d'$.
	We define the total grade value as
	\begin{equation}
		\mathrm{gr}:=4(u+u')+3(t+t')+3(x+x')+2(v+v') 
	\end{equation}
	and proceed by induction on it.
	The first dichotomy we have is: either $u+u'>0$ or $u+u'=0$. In the former case, we are again in a situation similar to the \textit{Case 1} of the proof of \cite[Theorem 6]{CTW23} as the only ingredients we need are the weak symmetry of $T(z)$, the fact that $L_{-1}\Omega=0$ and that the $N2$ commutation relations \eqref{eq:N2_cr} ``lower the grade''. Instead, for the latter case, we have either $t+t'>0$ or $t+t'=0$. The case $t+t'>0$ can be treated as in the proof of Theorem \ref{theo:unitarity_BS_representation_NS}, relaying on the weak symmetry of $G^1(z)$ and the fact that $G^1_{-\half}\Omega=0$ to define an operator as in \eqref{eq:operator_A_for_superconformal}, and working with the $N2$ commutation relations \eqref{eq:N2_cr}. Instead for $t+t'=0$, we split again the problem into the two cases: $x+x'>0$ and $x+x'=0$. The former is basically equal to the previous case $t+t'>0$, whereas the latter is the easier one as we can rely on the fact that $J(z)$ is symmetric. Indeed, 
	\begin{equation}
		(\Psi| \Psi')
		=( J_{-n_1}\cdots J_{-n_v}\Omega | J_{-n_1'}\cdots J_{-n'_{v'}}\Omega)
		=( J_{-n_2}\cdots J_{-n_v}\Omega | J_{n_1}J_{-n_1'}\cdots J_{-n'_{v'}}\Omega) \,.
	\end{equation}
	Using the commutation relations for the Fourier coefficients of $J(z)$, see \eqref{eq:N2_cr}, $J_{n_1}\Psi'$ is equal to $0$ if $n_1\not=n'_{l'}$ for all $l\in\{1,\dots, v'\}$ or to a finite sum of vectors of type $J\cdots J\Omega$ with negative indexes and a number of Fourier coefficients of type $J$ strictly smaller than $v'$. In this second case, as $d\not =d'$, we can apply the inductive hypothesis to conclude.
	
	Once we have proved that $L_0^\dagger=L_0$, the remaining part is again a trivial adaptation of what has been already done in the last part of the proof of Theorem \ref{theo:unitarity_BS_representation_NS}. 
\end{proof}

Therefore, we also get an alternative proof to the one given in \cite{ET88b} for the following result:

\begin{cor} \label{cor:unitarity_vacuum_representations_N2}
	The irreducible lowest weight representations of the $N=2$ super-Virasoro algebra $N2$ with central charge $c\geq 3$ and lowest weight $(h,q)=(0,0)$ are unitary.
\end{cor}

Now, as in the last part of Subsection \ref{subsubsec:step1}, we apply transformations of type $\varphi_{\kappa,\eta}$ for all $\kappa,\eta\in\R$ to the field $J^-(z)$ in \eqref{eq:bosonsxfermions_representations_N2_BS} to obtain a family of unitary lowest weight representations of $N2$ on $V$ with central charge $c(\kappa)=3+12\kappa^2$ and lowest weight vector $\Omega$ of lowest weight  $(h,q)=(\frac{\kappa^2+\eta^2}{2},0)$ given by:
\begin{equation}  \label{eq:bosonsxfermions_representations_N2_unitary_higher_zero_charge}
	\begin{split}
		T_{\kappa,\eta}(z)
			&:= 
		T(z)+\eta J^-(z)+\kappa (J^-)'(z)+\frac{\kappa^2+\eta^2}{2}\\
		G^1_{\kappa,\eta}(z)
			&:=
		G^1(z)-\frac{\eta}{\sqrt{2}}\Phi^-(z)-\sqrt{2}\kappa(\Phi^-)'(z)\\
		G^2_{\kappa,\eta}(z)
			&:=
		G^2(z)+\frac{\eta}{\sqrt{2}}\Phi^+(z)+\sqrt{2}\kappa(\Phi^+)'(z)\\
		J_\kappa(z)
			&:=
		J(z)+2\kappa J^+(z) \,.
	\end{split}
\end{equation}
Actually, we can go further and apply transformations of type $\varphi_{0,\omega}$ for all $\omega\in\R$ to the field $J^+(z)$ in \eqref{eq:bosonsxfermions_representations_N2_unitary_higher_zero_charge} to obtain a family of unitary lowest weight representations of $N2$ on $V$ with central charge $c(\kappa)=3+12\kappa^2$ and lowest weight vector $\Omega$ of lowest weight $(h,q)=(\frac{\kappa^2+\eta^2+\omega^2}{2},2\kappa\omega)$ given by:
\begin{equation}  \label{eq:bosonsxfermions_representations_N2_unitary_higher}
	\begin{split}
		T_{\kappa,\eta}(z)
		&:= 
		T(z)+\omega J^+(z)+\eta J^-(z)+\kappa (J^-)'(z)+\frac{\kappa^2+\eta^2+\omega^2}{2}\\
		G^1_{\kappa,\eta}(z)
		&:=
		G^1(z)+\frac{\omega}{\sqrt{2}}\Phi^+(z)-\frac{\eta}{\sqrt{2}}\Phi^-(z)-\sqrt{2}\kappa(\Phi^-)'(z)\\
		G^2_{\kappa,\eta}(z)
		&:=
		G^2(z)+\frac{\omega}{\sqrt{2}}\Phi^-(z)+\frac{\eta}{\sqrt{2}}\Phi^+(z)+\sqrt{2}\kappa(\Phi^+)'(z)\\
		J_\kappa(z)
		&:=
		J(z)+2\kappa J^+(z) +2\kappa\omega \,.
	\end{split}
\end{equation}

At this point, it remains to integrate the above representations of $N2$ to representations of the corresponding $N=2$ super-Virasoro net $\mathrm{SVir2}_{c(\kappa)}$. One can check that this can be done adapting the general argument in Subsection \ref{subsubsec:step2}, which indeed do not relay on the characteristics of the models, as instead it happened for the argument in Subsection \ref{subsubsec:step1}. 
The first step is to adapt Lemma \ref{lem:properties_smeared_Tkappa_Gkappa} to the $N=2$ setting.
Then the strategy is again to obtain solitons of $\mathrm{SVir2}_{c(\kappa)}$ using Weyl operators to implement the BMT representations corresponding to the transformations $\varphi_{\kappa,\eta}$, $\kappa,\eta\in\R$, applied to $J^-(z)$. Then, those solitons can be extended to representations of $\mathrm{SVir2}_{c(\kappa)}$ by Proposition \ref{prop:extending_solitons}. At this point, we can go further on the same line implementing the BMT representations corresponding to the transformations $\varphi_{0,\omega}$, $\omega\in\R$, applied to $J^+(z)$.  
Therefore, we can state the following theorem:

\begin{theo} \label{theo:irreprs_N=2_net}
Let $c\geq 3$. Then for all real numbers $h\geq \frac{c-3}{24}$ and all real numbers $q$ such that $\abs{q}\leq\sqrt{\frac{c-3}{3}\left(2h-\frac{c-3}{12}\right)}$, there exists an irreducible representation $\pi^{h,q}$ of the $N=2$ super-Virasoro net $\mathrm{SVir2}_c$ of lowest weight $(h,q)$.
\end{theo}

\begin{rem}
	Standard arguments show the existence of localized endomorphisms unitarily equivalent to the representations constructed in Theorem \ref{theo:irreprs_N=1_net} and in Theorem \ref{theo:irreprs_N=2_net}, see \cite[Proposition 14]{CKL08}. However, these abstract arguments do not provide an explicit expressions for such endomorphisms.
	A similar situation occurs for the Virasoro nets with $c>1$, cf.\ the last paragraph of \cite{BS90}.
	As in the Virasoro case, one may hope to construct these endomorphisms also explicitly in view of the concrete definitions of the corresponding representations. 
\end{rem}

\subsection{\texorpdfstring{$N=3$ super-Virasoro nets}{N=3 super-Virasoro nets}}  \label{subsec:N=3_super-virasoro_nets}
Let $F$ be the real free fermion VOSA as in Example \ref{ex:real_free_fermion}.
For all $k\not=-h^\vee=-\half$, the graded tensor product $W_k(\g,f)\hat{\otimes} F$ with $\g=\mathfrak{spo}(2|3)$ gives rise to the (simple) $N=3$ \textit{super-Virasoro VOSAs} $V^{c'(k)}(N3)$ with central charge $c'(k):=c(k)+\half=-(6k+3)$. 
The latter VOSAs can be constructed from the $N=3$ \textit{Lie superalgebra} $N3$, see \cite[Section 8.5]{KW04}, cf.\ \cite[Eq.s (14)--(15)]{GS88}. Furthermore, they are strongly graded-local whenever $k$ is in its unitary range $\frac{1}{4}\Z_{\leq -3}$ as described in Section \ref{section:strong_graded_locality} and references therein.
Then $c(k)$ will be in $\frac{3}{2}\Zplus$ and we will write $\mathrm{SVir3}_{c'(k)}$ for the corresponding $N=3$ \textit{super-Virasoro nets} $\A_{V^{c'(k)}(N3)}\cong\A_{W_k(\g,f)}\hat{\otimes} \F$. 
A full classification of unitary irreducible highest weight $W_k(\g,f)$-modules is provided in \cite[Section 13.1]{KMP23}, see also \cite[Theorem 6.9]{AKMP24}, cf.\ \cite{Mik90}.

Now, we construct some representations of $\mathrm{SVir3}_c$ with $c\geq\frac{3}{2}$, following the strategy in \cite[pp.\ 211--212]{GS88}.
Recall from Section \ref{section:strong_graded_locality} and references therein that $k=-\frac{3}{4}$ is a collapsing level such that $W_{-\frac{3}{4}}(\g,f)$ is isomorphic to the simple affine VOA $V_1(\mathfrak{sl}_2)$ with $c(-\frac{3}{4})=1$. It follows that $W_{-\frac{3}{4}}(\g,f)$ as well as $V^\frac{3}{2}(N3)$ is strongly rational. Moreover, the irreducible graded-local conformal net $\mathrm{SVir3}_\frac{3}{2}= \A_{V_1(\mathfrak{sl}_2)}\hat{\otimes} \F$ is completely rational as $\F$ is completely rational. This follows from e.g.\ the fact that the Bose subnet $\F^\Gamma=\Virnet_\half$ of $\F$ is completely rational. 
As far as $V^3(N3)$ is concerned, it follows from \cite[Eq.s (18)--(19)]{GS88}, based on \cite[Section 4]{Sch87}, that it can be identified with a unitary subalgebra of the simple unitary VOSA $V:=M(1)\hat{\otimes} F^{\hat{\otimes} 4}$ with the same conformal vector. By \cite[Theorem 6.1]{CGH}, $\mathrm{SVir3}_3$ is a covariant subnet of $\A_V$. 
For any positive $\eta\in\R$, consider a BMT representation, see page \pageref{bmt_representations} and references therein, $\alpha^\eta$ of $\A_{M(1)}$ such that the infimum of the spectrum of $L_0^{\alpha^\eta}$, as defined in Remark \ref{rem:irrationality_local_net_from_graded_net}, is $h_{\alpha^\eta}=\frac{\eta^2}{2}$.
Set $\tilde{\pi}^\eta$ as the irreducible diffeomorphism covariant representation of $\A_V$ given by the tensor product of $\alpha^\eta$ and the vacuum representation of $\F^{\hat{\otimes}4}$. Note that $h_{\tilde{\pi}^\eta}=h_{\alpha^\eta}$. 
Then $\tilde{\pi}^\eta$ restricts to a diffeomorphism covariant representation $\pi^\eta$ of $\mathrm{SVir3}_3$. Moreover, $\A_V$ and $\mathrm{SVir3}_3$ have the same central charge, so that $\mathrm{SVir3}_3$ contains the Virasoro subnet of $\A_V$. 
From Proposition \ref{prop:diff_cov_grade-local_net}, $h_{\pi^\eta}=h_{\alpha^\eta}$ and thus we have a family $\mathscr{E}:=\{\pi^\eta\mid \eta\in\R\}$ of representations of $\mathrm{SVir3}_3$ such that the set $\{h_\pi\mid \pi\in\mathscr{E}\}$ is infinite. Therefore, $\mathrm{SVir3}_3$ is not completely rational.

It follows from the $N3$ commutation relations that for all $c\in \frac{3}{2}\Z_{\geq 3}$, we can realize the VOSA inclusion
\begin{equation} \label{eq:N3_inclusions}
	V^c(N3)\hookrightarrow \left\{
	\begin{array}{lr}
		V^3(N3)^{\hat{\otimes}n} & \mbox{ if } c=3n \,,\,\, n\in \Zplus\\
		V^\frac{3}{2}(N3)\hat{\otimes}V^3(N3)^{\hat{\otimes}\frac{n-1}{2}} & \mbox{ if } c=\frac{3}{2}n \,,\,\, n\in \Z_{>1} \,,\,\, n  \mbox{ odd }
	\end{array}
	\right. 
\end{equation}
which is also unitary.
By \cite[Theorem 6.1]{CGH}, we have the corresponding inclusions in the graded-local conformal net setting.
Therefore, we can obtain a family of representations as in Remark \ref{rem:irrationality_local_net_from_graded_net} for $\mathrm{SVir3}_c$ with $c>3$ just considering the graded tensor product of a suitable number of representations of $\mathrm{SVir3}_\frac{3}{2}$ and of $\mathrm{SVir3}_3$ according to \eqref{eq:N3_inclusions}. This argument shows that $\mathrm{SVir3}_c$ with $c\geq3$ is not completely rational. In turn, we also have that  $\A_{W_k(\g,f)}$ with $k\in\frac{1}{4}\Z_{< -3}$ is not completely rational.

\subsection{\texorpdfstring{$N=4$ super-Virasoro nets}{N=4 super-Virasoro nets}}  \label{subsec:N=4_super-virasoro_nets}

The simple minimal $W$-algebras $W_k(\g,f)$ with $\g=\mathfrak{psl}(2|2)$ are the $N=4$ \textit{super-Virasoro VOSAs} $V^{c(k)}(N4)$ with central charges $c(k)=-6(k+1)$, arising from the $N=4$ \textit{Lie superalgebra} $N4$, see \cite[Section 8.4]{KW04}. 
Let $k$ be in the corresponding unitary range, that is $k\in\Z_{\leq -2}$, so that $c(k)\in 6\Zplus$, see Section \ref{section:strong_graded_locality} and references therein.
The irreducible unitary representations of this models have been classified in \cite[Section 13.2]{KMP23}, see also \cite[Theorem 6.9]{AKMP24}, cf.\ \cite{ET87, ET88, ET88b}, showing the non-rationality of these VOSAs.
To show that the $N=4$ \textit{super-Virasoro nets} $\mathrm{SVir4}_{c(k)}:=\A_{V^{c(k)}(N4)}$ are not completely rational, we can construct a family of representations as in Remark \ref{rem:irrationality_local_net_from_graded_net}, miming what done for the case of the $N=3$ super-Virasoro net $\mathrm{SVir3}_3$ in Section \ref{subsec:N=3_super-virasoro_nets}. 
The only ingredient we need is to identify for all $c(k)\in 6\Zplus$, or equivalently for all $k\in\Z_{\leq -2}$, $V^{c(k)}(N4)$ with a unitary subalgebra of a well-known VOSA $U^{c(k)}$ with the same conformal vector as done in \eqref{eq:N3_inclusions} for the $N=3$ Lie superalgebra case. 
Indeed, $U^{c(k)}$ is the graded tensor product of $-(k+1)$ copies of $(M(1)\hat{\otimes} F)^{\hat{\otimes}4}$, see \cite[Eq.\ (2)]{ET87}, cf.\ also \cite[Section 13.2]{KMP23}, where $M(1)$ is the Heisenberg VOA and $F$ is the real free fermion VOSA, see references in Remark \ref{rem:unitary_range_and_collapsing} and Example \ref{ex:real_free_fermion} respectively.

\subsection{\texorpdfstring{Big $N=4$ super-Virasoro nets}{Big N=4 super-Virasoro nets}}
\label{subsec:big_N=4_super-virasoro_nets}

Let $\g=D(2,1;a)$ with $a\in\C\backslash\{-1,0\}$. For all $k\not=-h^\vee=0$, tensorizing the simple minimal $W$-algebra $W_k(\g,f)$ with the Heisenberg VOA $M(1)$ and four copies of the real free fermion VOSA $F$, see references in Remark \ref{rem:unitary_range_and_collapsing} and Example \ref{ex:real_free_fermion} respectively,  we get the \textit{big $N=4$ super-Virasoro algebras}, see \cite[Section 8.6]{KW04}, cf.\ \cite{STV88} and references therein. Denote it by $V^{c'(k)}$, where $c'(k):=c(k)+3=-6k$ is the central charge.
Let $k$ be in the corresponding unitary range as described in Section \ref{section:strong_graded_locality} and references therein. 
Recall also that for all $m\in \Z_{\geq 2}$, $W_{\frac{-m}{m+1}}(\g,f)$ collapses to the simple affine VOA $V_{m-1}(\mathfrak{sl}_2)$ with $c(\frac{-m}{m+1})=\frac{3(m-1)}{m+1}$.
Therefore $\A_{W_{\frac{-m}{m+1}}(\g,f)}$ is isomorphic to $\A_{V_{m-1}(\mathfrak{sl}_2)}$, which is completely rational. 
For $k$ non-collapsing, $W_k(\g,f)$ is not rational by \cite[Theorem 6.9]{AKMP24}, whereas $V^{c'(k)}$ is not rational for any $k$ in the unitary range as $M(1)$ is not.
We know that $V^{c'(k)}$ is strongly graded-local.
Therefore, similarly to Sections \ref{subsec:N=3_super-virasoro_nets}--\ref{subsec:N=4_super-virasoro_nets}, taking tensor products of BMT representations of $\A_{M(1)}$ with the vacuum ones of $\A_{W_k(\g,f)}$ and of $F^{\hat{\otimes}4}$, we can build a family of irreducible representations as in Remark \ref{rem:irrationality_local_net_from_graded_net} of the \textit{big $N=4$ super-Virasoro nets} $\A_{V^{c'(k)}}=\A_{W_k(\g,f)}\hat{\otimes} \A_{M(1)}\hat{\otimes} \F^{\hat{\otimes}4}$, so concluding that $\A_{V^{c'(k)}}$ is not completely rational.

\bigskip
\noindent
{\small
{\bf Acknowledgements.} We would like to thank Pierluigi M\"{o}seneder Frajria and Paolo Papi for stimulating discussions and detailed explanations on the results in \cite{KMP23}. We acknowledge support from the \emph{MIUR Excellence Department Project MatMod@TOV} awarded to the Department of Mathematics, University of Rome ``Tor Vergata'', CUP E83C23000330006 and from the University of Rome ``Tor Vergata'' funding OANGQS, CUP E83C25000580005. S.C.\ also acknowledge support from the GNAMPA-INDAM project {\it Operator algebras and infinite quantum systems}, CUP E53C23001670001.
Part of this work was done while T.G.\ was Senior Research Associate, Leverhulme Trust at the School of Mathematical Sciences, Lancaster University, supported by the Leverhulme Trust Research Project Grant RPG-2021-129.

\bigskip
\noindent
{\small
\emph{Data sharing is not applicable to this article as no new data were created or analyzed in this study.} }

\bigskip
\noindent
{\small
\emph{On behalf of all authors, the corresponding author states that there is no conflict of interest.} }

\end{document}